\documentclass[11pt]{article}

\usepackage{setspace}
\usepackage{graphicx}
\usepackage{latexsym}
\usepackage{times}
\usepackage{fullpage}
\usepackage{dsfont}
\usepackage{color}
\usepackage{lineno}
\usepackage{enumerate,authblk}
\usepackage{tikz}

\usepackage{amsmath,amsfonts,amssymb,amsthm}
\usepackage{xspace}
\usepackage[linesnumbered,boxed,ruled, vlined]{algorithm2e} 
\usepackage{mathtools}

\long\def\ignore#1{}

\newtheorem{theorem}{Theorem}
\newtheorem{lemma}[theorem]{Lemma}
\newtheorem{corollary}[theorem]{Corollary}

\newtheorem{definition}[theorem]{Definition}

\newtheorem{observation}[theorem]{Observation}

\makeatletter
\newtheorem*{rep@theorem}{\rep@title}
\newcommand{\newreptheorem}[2]{%
\newenvironment{rep#1}[1]{%
 \def\rep@title{#2 \textbf{\ref{##1}}}%
 \begin{rep@theorem}}%
 {\end{rep@theorem}}}
\makeatother
\newreptheorem{theorem}{\textbf{Theorem}}

\newcommand{\commentout}[1]{}
\newcommand{\eat}[1]{}

\newcommand{\calH}{{\mathcal H}}

\newcommand{\calA}{{\mathcal A}}
\newcommand{\calC}{{\mathcal C}}

\newcommand{\calL}{{\mathcal L}}

\newcommand{\calF}{{\mathcal F}}
\newcommand{\calP}{{\mathcal P}}
\newcommand{\calK}{{\mathcal K}}
\newcommand{\bcalK}{\overline{\mathcal K}}
\newcommand{\calS}{{\mathcal S}}

\newcommand{\calR}{{\mathcal R}}
\newcommand{\calE}{{\mathcal E}}

\newcommand{\calZ}{{\mathcal Z}}

\newcommand{\alg}{\mathcal{E}}

\newcommand{\flatdist}{\mathrm{J}}

\newcommand{\Prob}{{\mathrm{Pr}}}
\renewcommand{\Pr}{{\mathrm{Pr}}}
\newcommand{\Exp}{{\mathbb{E}}}

\newcommand{\E}{{\mathbb{E}}}

\newcommand{\R}{{\mathbb{R}}}
\newcommand{\boldS}{{\mathbf{S}}}
\newcommand{\boldR}{{\mathbb{U}}}

\newcommand{\dist}{\mathrm{d}}
\renewcommand{\d}{\mathrm{d}}
\newcommand{\maxdist}{\mathrm{K}}
\newcommand{\setmaxdist}{\mathrm{K}}

\newcommand{\fun}{\mathsf{dist}}

\newcommand{\cost}{\mathsf{cost}}

\newcommand{\polar}{\star}
\newcommand{\perror}{\tau}

\newcommand{\ACORESET}{$\mathbb{A}$}

\newcommand{\bH}{\overline{H}}

\newcommand{\tO}{\widetilde{O}}

\newcommand{\set}[1]{\left\{#1\right\}}
\newcommand{\tuple}[1]{\left(#1\right)}

\newcommand{\e}{\epsilon}

\newcommand{\p}{p}
\renewcommand{\P}{\mathcal{P}}

\newcommand{\poly}{\mathrm{poly}}

\newcommand{\exprkernel}{$(\e,r)$-\textsc{fpow-kernel}}
\newcommand{\probkernel}{$(\eps,\perror)$-\textsc{quant-kernel}}

\newcommand{\dw}{\omega}   
\newcommand{\innerprod}[2]{\langle #1,#2 \rangle}

\newcommand{\eps}{\varepsilon}
\renewcommand{\epsilon}{\varepsilon}

\newcommand{\topic}[1]{\vspace{0.2cm}\noindent {\bf #1}}
\newcommand{\jian}[1]{{\color{red} $\langle${\sffamily\small Jian: }#1$\rangle$}}

\newcommand{\kcentercoreset}{\textsc{SKC-Coreset}}
\newcommand{\jflatcoreset}{\textsc{SJFC-Coreset}}

\usepackage[top=1in, bottom=1in, left=1in, right=1in]{geometry}

\title{Stochastic $k$-Center and $j$-Flat-Center Problems}

\author{Lingxiao Huang \quad\quad\quad\quad Jian Li \\
Institute for Interdisciplinary Information Sciences\\
Tsinghua University, China
}

\begin{document}

\pagenumbering{gobble}
\begin{titlepage}

\maketitle

\begin{abstract}
Solving geometric optimization problems over uncertain data
has become increasingly important in many applications and
has attracted a lot of attentions in recent years.
In this paper, we study two important geometric optimization problems,
the $k$-center problem and the $j$-flat-center problem, over stochastic/uncertain data points in Euclidean spaces.
For the stochastic $k$-center problem, we would like to
find $k$ points in a fixed dimensional Euclidean space,
such that the expected value of the $k$-center objective is minimized.
For the stochastic $j$-flat-center problem,
we seek a $j$-flat (i.e., a $j$-dimensional affine subspace)
such that the expected value of the maximum distance from any point
to the $j$-flat is minimized.
We consider both problems under two popular stochastic geometric models,
the existential uncertainty model, where the existence of each point
may be uncertain, and the locational uncertainty model, where
the location of each point may be uncertain.
We provide the first PTAS (Polynomial Time Approximation Scheme) for both problems under the two models.
Our results generalize the previous results for stochastic minimum enclosing
ball and stochastic enclosing cylinder.
\end{abstract}
\end{titlepage}

\newpage

\pagenumbering{arabic}
\setcounter{page}{1}

\section{Introduction}

With the prevalence of automatic information extraction/integration systems, and predictive machine learning algorithms in numerous application areas,
we are faced with a huge volume of data which is inherently uncertain and noisy. The most principled way for managing, analyzing and optimizing over such uncertain data is to use stochastic models (i.e., use probability
distributions over possible realizations to capture the uncertainty).
This has led to a surge of interests in stochastic combinatorial
and geometric optimization problems in recently years
from several research communities
including theoretical computer science, databases, machine learning.
In this paper, we study two classic geometric optimization problems,
the $k$-center problem and the $j$-flat center problem in Euclidean spaces.
Both problems are important in geometric data analysis.
We generalize both problems to the stochastic settings.
We first introduce the stochastic geometry models, and
then formally define our problems.

\topic{Stochastic Geometry Models:}
There are two natural and popular stochastic geometry models, under which most
of stochastic geometric optimization problems
are studied, such as closest pairs \cite{KCS11b}, nearest neighbors \cite{agarwal2012nearest,KCS11b}, minimum spanning trees \cite{huang2015approximating,kamousi2011stochastic}, perfect matchings \cite{huang2015approximating}, clustering \cite{cormode2008approximation,guha2009exceeding}, minimum enclosing balls \cite{enclosingball14}, and
range queries \cite{ADP13,agarwal2012range,li2014range}.
We define them formally as follows:

\begin{enumerate}
\item Existential uncertainty model:
Given a set $\P$ of $n$ points
in $\R^d$,
each point $s_i \in \P$ ($1\leq i\leq n$) is associated with a real number (called {\em existential probability})
$\p_{i}\in [0,1]$, i.e., point $u_i$ is present independently with probability $p_i$. A realization $P\sim \calP$ is a point set
which is realized with probability $\Prob[\vDash P]=\prod_{s_i\in P}p_i \prod_{s_i\notin P}(1-p_i)$.
\item
Locational uncertainty model:
Assume that there is a set $\P$ of $n$ nodes and the existence of each node is certain.
However, the location of each node $u_i\in\P$ ($1\leq i\leq n$) might be a random point in $\R^d$.
We assume that the probability distribution for each $u_i\in \calP$ is discrete and independent of other points.
For a node $u_i\in \P$ and a point $s_j\in \R^d$ $(1\leq j\leq m)$, we define $\p_{i,j}$ to be the probability that the location of node $u_i$ is $s_j$.
\end{enumerate}


\topic{Stochastic $k$-Center:}
The deterministic Euclidean $k$-center problem
is a central problem in geometric optimization~\cite{agarwal2002exact,agarwal2005geometric}.
It asks for a $k$-point set $F$ in $\R^d$ such that
the maximium distance from any of the $n$ given points to
its closest point in $F$ is minimized.
Its stochastic version is naturally motivated:
Suppose we want to build $k$ facilities to serve a set of uncertain
demand points, and our goal is to minimize the expectation of the
maximum distance from any realized demand point to its closest facility.

\begin{definition}
\label{def:kcenter}
For a set of points $P\in \R^d$, and a $k$-point set
$F=\{f_1,\ldots,f_k)\mid f_i\in\R^d,1\leq i\leq k\}$, we define $\maxdist(P,F)=\max_{s\in P}\min_{1\leq i\leq k}\dist(s,f_i)$
as the $k$-center value of $F$ w.r.t. $P$.
We use $\calF$ to denote the family of all $k$-point sets in $\R^d$.
Given a set $\calP$ of $n$ stochastic points (in either the existential or locational uncertainty model) in $\R^d$, and a $k$-point set $F\in \calF$,
we define the expected $k$-center value of $F$ w.r.t $\calP$ as
$$
\maxdist(\calP,F)=\Exp_{P\sim \calP}[\maxdist(P,F)].
$$
In the stochastic minimum $k$-center problem,
our goal is to find a $k$-point set $F\in \calF$ which minimizes $\maxdist(\calP,F)$.
In this paper, we assume that
both the dimensionality $d$ and $k$ are fixed constants.
\end{definition}

\topic{Stochastic $j$-Flat-Center:}
The deterministic $j$-flat-center problem
is defined as follows:
given $n$ points in $\R^d$, we
would like to find a $j$-flat $F$ (i.e., a $j$-dimensional affine subspace)
such that the maximum distance from any given point to $F$ is minimized.
It is a common generalization of the minimum enclosing ball ($j=0$),
minimum enclosing cylinder ($j=1$), and minimum width problems ($j=d-1$),
and has been well studied in computational geometry~\cite{agarwal2005geometric,FL11,varadarajan2012sensitivity}.
Its stochastic version is also naturally motivated by the stochastic variant
of the $\ell_\infty$ regression problem:
Suppose we would like to fit a set of points by an affine subspace.
However, those points may be produced by some machine learning algorithm,
which associates some confidence level to each point
(i.e., each point has an existential probability).
This naturally gives rise to the stochastic $j$-flat-center problem.
Formally, it is defined as follows.

\begin{definition}
\label{def:jflat}
Given a set $P$ of $n$ points in $\R^d$, and a $j$-flat $F\in \calF$ ($0\leq j\leq d-1$), where $\calF$ is the family of all $j$-flats in $\R^d$,
we define the $j$-flat-center value of $F$ w.r.t. $P$ to be
$
\flatdist(P,F)=\max_{s\in P}\dist(s,F),
$
where $\dist(s,F)=\min_{f\in F}\dist(s,f)$ is the distance between point $s$ and $j$-flat $F$.
Given a set $\calP$ of $n$ stochastic points (in either the existential or locational model) in $\R^d$, and a $j$-flat $F\in \calF$ ($0\leq j\leq d-1$), we define the expected $j$-flat-center value of $F$ w.r.t. $\calP$ to be
$$
\flatdist(\calP,F)=\Exp_{P\sim \calP}[\flatdist(P,F)].
$$
In the stochastic minimum $j$-flat-center problem,
our goal is to find a $j$-flat $F$ which minimizes
$\flatdist(\calP,F)$.
\end{definition}


\subsection{Previous Results and Our contributions}

Recall that a {\em polynomial time approximation scheme (PTAS)} for a minimization problem is an algorithm $A$ that
produces a solution whose cost is at most $1+\e$ times the optimal cost
in polynomial time, for any fixed constant $\e>0$.

\topic{Stochastic $k$-Center:}
Cormode and McGregor~\cite{cormode2008approximation} first studied
the stochastic $k$-center problem in a finite metric graph
under the locational uncertainty model,
and obtained a bi-criterion constant approximation.
Guha and Munagala~\cite{guha2009exceeding} improved their result to a single-criterion constant factor approximation.
Recently, Wang and Zhang~\cite{wang2015one}
studied the stochastic $k$-center problem on a line,
and proposed an efficient exact algorithm.
No result better than a constant approximation
is known for the Euclidean space $\R^d$ ($d\geq 2$).
We obtain the first PTAS for the stochastic $k$-center problem in $\R^d$.

\begin{theorem}
\label{thm:praskcenter}
Assume that both $k$ and $d$ are fixed constants.
There exists a PTAS for the stochastic minimum $k$-center problem
in $\R^d$,
under either the existential or the locational uncertainty model.
\end{theorem}

Our result generalizes the PTAS for stochastic minimum enclosing ball
by Munteanu et al.~\cite{enclosingball14}.
We remark that the assumption that $k$ is a constant is necessary for
getting a PTAS, since even the deterministic Euclidean $k$-center problem
is APX-hard for arbitrary $k$ even in $\R^2$~\cite{feder1988optimal}.

\topic{Stochastic $j$-Flat-Center:}
Our main result for the stochastic $j$-flat-center is as follows.

\begin{theorem}
	\label{thm:prasjflat}
	Assume that the dimensionality $d$ is a constant.
	There exists a PTAS for the stochastic minimum $j$-flat-center problem,
	under either the existential or the locational uncertainty model.
\end{theorem}

This result also generalizes the PTAS for stochastic minimum enclosing ball (i.e., 0-flat-center) by Munteanu et al.~\cite{enclosingball14}.
It also generalizes a previous PTAS for the stochastic minimum enclosing cylinder (i.e., 1-flat-center) problem in the existential model where
the existential probability of each point is assumed to be lower bounded
by a small fixed constant~\cite{huang2014epsilon}.

\eat{
	Recently, Huang et al.~\cite{huang2014epsilon} generalized
	the notion of $\epsilon$-kernel coreset to stochastic points.
	In particular, they showed that,
	for a set $\calP$ of stochastic points,
	the expected directional width
	$\E_{P\sim \calP}[\max_{p\in P}
	\innerprod{\omega}{p}-\min_{p\in P}\innerprod{\omega}{p}]$ for every direction $\omega$ is equal to the directional width of a single deterministic convex set, and they provided an efficient algorithm
	for constructing an $\epsilon$-kernel (of constant size) for the convex set. Using the $\epsilon$-kernel, one can find in constant time
	the direction such that the expected directional width of $\calP$ is minimized, which is equivalent to the $(d-1)$-flat center
}

\topic{Our techniques:}
Our techniques for both problems heavily rely
on the powerful notion of coresets.
In a typical deterministic geometric optimization problem,
an instance $P$ is a set of deterministic (weighted) points.
A coreset $S$ of $P$ is a set of (weighted) points, such that the solution for the optimization problem over $S$ is a good
approximate solution for $P$.
\footnote{It is possible to define coresets for other classes of optimization
	problems.}
Recently, Huang et al.~\cite{huang2014epsilon} generalized
the notion of $\epsilon$-kernel coreset (for directional width)
to stochastic points.
However, their technique can only handle directional width, and extending
it to problems such as stochastic minimum enclosing cylinder  requires
certain technical assumption (see~\cite{huang2014epsilon} for the detailed
discussion).

In this paper, we introduce a new framework for
solving geometric optimization problems over stochastic points.
For a stochastic instance $\calP$, we consider $\calP$ as a collection of realizations $\calP=\{P\mid P\sim \calP\}$. Each realization $P$ has a weight $\Prob[\vDash P]$, which is its realized probability.
Now, we can think the stochastic problem as
a certain deterministic problem over (exponential many) all realizations
(each being a point set).
Our framework constructs an object $\calS$ satisfying the following properties.

\begin{enumerate}
	\item Basically, $\calS$ has a constant size description
	(the constants may depend on $d$, $\epsilon$, and $k$).
	\item
	The objective value for a certain deterministic optimization problem
	over $\calS$ can approximate the objective for the original
	stochastic problem well.
	Moreover, the solution to the deterministic optimization
	over $\calS$ is a good approximation for the original problem
	as well.
\end{enumerate}

In a high level, $\calS$ serves very similar roles as the coresets
in the deterministic setting.
Note that the form of $\calS$ may vary for different problems:
in stochastic $k$-center, it is a collection of weighted point sets (we call $\calS$ an \kcentercoreset);
in stochastic $j$-flat-center, it is a combination of two collections of weighted point sets
for two intermediate problems
(we call $\calS$ an \jflatcoreset).

For stochastic $k$-center under the existential model,
we construct an \kcentercoreset\ $\calS$
in two steps.
First,
we map all realizations to their
additive $\e$-coresets
(for deterministic $k$-centers)~\cite{agarwal2002exact}.
Since there are only a polynomial number of possible
additive $\e$-coresets,
the above mapping can partition the space of all realizations into
a polynomial number of parts,
such that the realizations in each part
have very similar objective functions.
Moreover, for each additive $\e$-coresets,
it is possible to compute the total probability of
the realizations that are mapped to the coreset.
In fact, this requires a subtle modification of
the construction in~\cite{agarwal2002exact}
so that we can compute the aforementioned probability
efficiently.
This step has reduced the exponential number of realizations
to a polynomial size representation.
Next, we define a generalized shape fitting problem, call
the {\em generalized $k$-median} problem,
over the collection of above additive $\e$-coresets.
Then, we need to properly generalize the previous definition of
coreset and the total sensitivity (a notion proposed in
the deterministic coreset context by Langberg and Schulman \cite{langberg2010universal}),
and prove a constant upper bound
for the generalized total sensitivity
by relating it to the total sensitivity of the ordinary
$k$-median problem.
The \kcentercoreset\ $\calS$ is a generalized coreset
for the generalized $k$-median problem, which consists
of a constant number of weighted point sets.

For stochastic $k$-center under
the locational model, computing the weight
for each set in
the \kcentercoreset\ $\calS$ is somewhat more
complicated.
We need to reduce the computational problem to
a family of bipartite holant problems, and
apply the celebrated result by Jerrum, Sinclair, and
Vigoda~\cite{jerrum2004polynomial}.

For the stochastic minimum $j$-flat-center problem,
we proposed an efficient algorithm for constructing an \jflatcoreset. We utilize several ideas in the recent work~\cite{huang2014epsilon},
as well as prior results on the shape fitting problem.
We first partition the realizations $P\sim \calP$ into two parts through
a construction similar to the \probkernel\ construction in~\cite{huang2014epsilon}.
Roughly speaking, after linearization,
we need to find a convex set $\calK$ in a higher dimensional space
such that the total probability of any point falling outside
$\calK$ is small, but not so small such that
in each direction the expected directional width of $\calP$
is comparable to that of $\calK$.
Then, for those points inside $\calK$,
it is possible to use a slight modification of the construction in ~\cite{huang2014epsilon} to construct a collection of weighted point sets.
For the points outside $\calK$, since the total probability is small,
we reduce the problem to a weighted $j$-flat-median problem, and use the coreset in~\cite{varadarajan2012sensitivity}
(this step is similar to that in \cite{enclosingball14}).
By combining the two collections, we obtain
the \jflatcoreset\
$\calS$ for the problem, which is of constant size.
Then, we can easily obtain a PTAS by solving a constant size polynomial system
defined by $\calS$.

We remark that our overall approach is very different from that in
Munteanu et al.~\cite{enclosingball14} (except one aforementioned step
and that they also
crucially used some machinary from the coreset literature).
Munteanu et al.~\cite{enclosingball14} defined a near-metric distance measure
$m(A,B)=\max_{a\in A, b\in B}\dist(a,b)$
for two non-empty point sets $A,B$.
This near-metric measure satisfies many metric properties,
like non-negativity, symmetry and the triangle inequality.
By lifting the problem to the space defined by such metric and utilizing
a previous coreset result for clustering,
they obtained a PTAS  for the problem.
However, in the more general stochastic minimum $k$-center problem and stochastic minimum $j$-flat-center problem,
it is unclear how to translate the distance function between point sets and $k$-centers
or  point sets and $j$-flat sets to a near-metric distance (and still satisfies symmetry and triangle inequality).

\subsection{Other Related work}
Recently, Huang et al. \cite{huang2014epsilon}
generalized the notion of $\epsilon$-kernel coreset
in \cite{agarwal2004approximating}
to stochastic points and applied it to the stochastic minimum spherical shell, minimum enclosing cylinder and minimum cylindrical shell problems. However, the stochasticity
introduces certain complications in lifting the problems to higher
dimensional space and converting the solution back.
Hence, they could only obtain PTAS for those problems under
the assumption that the existential probability of each point is lower bounded by a small fixed constant.
Abdullah {\em et al.} \cite{ADP13} also studied coresets
for range queries over stochastic data.

Kamousi, Chan and Suri~\cite{kamousi2011stochastic} studied
the problem of estimating the expected length of several geometric objects,
such as MST, the nearest neighbor graph, the Gabriel graph and
the Delaunay triangulation in stochastic geometry models.
Huang and Li~\cite{huang2015approximating} considered several other problems including closest pair, diameter, minimum perfect matching, and minimum cycle cover. Many stochastic geometry problems have also been studied recently,
such as computing the expected volume of a set of probabilistic rectangles in a Euclidean space~\cite{yildiz2011union},
convex hulls \cite{agarwal2014convex}, and
skylines over probabilistic points~\cite{afshani2011approximate,atallah2011asymptotically}

\eat{If $k$ is a constant, the existence of an additive coreset of a constant size for $k$-line-center, i.e., for the problem of covering $P$ by $k$ congruent cylinders of the minimum radius, was first proved by Agarwal et al.~\cite{agarwal2002exact}.
Langberg and Schulman~\cite{langberg2010} showed that for the weighted $k$-median/$k$-means problem,
\footnote{The $k$-median/$k$-means problem in the existential uncertainty model can be considered as a weighted $k$-median/$k$-means problem.}
there exists an $\e$-coreset of size depending polynomially on $d$ and $k$ by bounding the total sensitivity. Varadarajan and Xiao~\cite{varadarajan2012sensitivity} studied the $k$-line clustering problem and the $(j,k)$ integer projective clustering problem, and showed that there exists an $\e$-coreset of a poly-logarithmic size.
}


For the deterministic $k$-center problem,
Gonzalez gave a 2-approximation greedy algorithm in metric space.
Hochbaum and Shmoys~\cite{hochbaum1986unified} showed that 2 is optimal
in general metric spaces unless $P=NP$.
In Euclidean spaces, the best hardness of approximation known is 1.82 even for $\R^2$~\cite{feder1988optimal}.
Agarwal and Procopiuc~\cite{agarwal2002exact} showed that there exists an additive coreset of a constant size if both $k$ and $d$ are constants.
Har-Peled and Varadarajan~\cite{har2002projective} studied the minimum enclosing cylinder (1-flat-center) problem in $\R^d$, and obtained a PTAS running in $dn^{(1/\e)^{O(1)}}$ time. Their algorithm can be extended to the $j$-flat-center problem, and obtained a PTAS running in $dn^{(j/\e)^{O(1)}}$ time. Badouiu, Clarkson and Panigrahy \cite{badoiu2003smaller,panigrahy2004minimum} improved their result of the $j$-flat-center problem to a linear-time PTAS.

Note that both the $k$-center and $j$-flat center problems
are special cases of the $\ell_\infty$ version of $(j,k)$-projective clustering problem, where we want to find $k$ $j$-flats to minimize
the maximum distance from any point to its closest $j$-flat.
\footnote{
	The minimum $k$-center problem is the $(0,k)$-projective clustering problem, and the minimum $j$-flat-center problem is the $(j,1)$-projective clustering problem.
	}
Har-Peled and Varadarajan~\cite{har2002projective} obtained
the first PTAS when both $j$ and $k$ are constants
($d$ can be arbitrary).

The $\ell_1$ version of the projective clustering problems
(with the corresponding coresets)
have also been studied extensively
(see e.g., \cite{FL11, feldman2013turning, varadarajan2012near, varadarajan2012sensitivity}).
In Euclidean space $\R^d$,
\eat{
Feldman and Langberg \cite{FL11} showed that for subspace approximation (i.e., $j$-flat median) problems, there exists a strong coreset of size $O(dj/\e^2)$ and a weak coreset of size $O(j^2\log (1/\e)/\e^3)$. Here, a strong coreset is a weighted set of points which can approximate the distance to every possible $j$-flat, and a weak coreset only satisfies that a $(1+\e)$-approximation for the optimal solution of the coreset yields a $(1+\e)$-approximation for the optimal solution of the original data set.}
Feldman and Langberg \cite{FL11} gave a coreset for the $k$-median problem, the subspace approximation (i.e., $j$-flat median) problem, and the $k$-line-median problem. Varadarajan et al. \cite{varadarajan2012sensitivity} also considered the $k$-line-median problem, and gave a coreset of size $O(k^{f(k)}d (\log n)^2/\e^2)$, where $f(k)$ is a function depending only on $k$.

\section{Preliminaries}
\label{sec:pre}

\paragraph{Generalized Shape Fitting Problems and Coresets}

As we mentioned in the introduction,
an \kcentercoreset\ $\calS$
is a collection of weighted point sets.
Hence, we need to define the generalized shape fitting problems,
which are defined over a collection of (weighted) point sets,
(recall that the traditional shape fitting problems
(see e.g., \cite{varadarajan2012sensitivity})
are defined over a set of (weighted) points).
We use $\R^d$ to denote the $d$-dimensional Euclidean space.
Let $\d(p, q)$ denote the Euclidean distance between point $p$
and $q$ and $\d(p,F)=\min_{q\in F}\d(p,q)$ for any $F\subset \R^d$.
Let $\boldR^d=\{P\mid P\subset \R^d, |P| \text{ is finite} \}$ be the collection of all finite discrete point sets in $\R^d$.

\begin{definition}
\label{def:shape}
(Generalized shape fitting problems)
A generalized shape fitting problem is specified by a triple $(\R^d,\calF,\fun)$. Here the set $\calF$ of shapes
is a family of subsets of $\R^d$
(e.g., all $k$-point sets, or all $j$-flats),
and $\fun: \boldR^d\times \calF\rightarrow \R^{\geq 0}$ is a generalized distance function, defined as
$\fun(P,F)=\max_{s\in P}\dist(s,F)$
for a point set $P\in \boldR^d$ and a shape $F\in \calF$.
\footnote{
	Note that $\fun$ may not be a metric in general.
	}
An instance $\boldS$ of the generalized shape fitting problem
is a (weighted) collection $\{S_1, \ldots, S_m\}$ ($S_i\in \boldR^d$)
of point sets, and each $S_i$ has a positive weight $w_i\in \R^+$.
For any shape $F\in \calF$,
define the total generalized distance from $\boldS$ to $F$ to be
$\fun(\boldS,F)=\sum_{S_i\in \boldS}w_i \cdot \fun(S_i,F)$.
Given an instance $\boldS$,
our goal is to find a shape $F\in \calF$, which minimizes the
total generalized distance $\fun(\boldS,F)$.
\end{definition}

If we replace $\boldR^d$ with $\R^d$, the above
definition reduces to the traditional shape fitting problem
defined in e.g., \cite{varadarajan2012sensitivity}.
Now, we define what is a coreset
for a generalized shape fitting problem.

\begin{definition}
\label{def:core}
(Generalized Coreset)
Given a (weighted) instance $\boldS$ of a generalized shape fitting problem $(\R^d,\calF,\fun)$
with a weight function $w:\boldS\rightarrow \R^+$,
a generalized $\e$-coreset of $\boldS$ is a (weighted) collection $\calS\subseteq \boldS$ of point sets, together with a weight function $w':\calS\rightarrow \R^+$, such that for any shape $F\in \calF$, we have that
$$ \sum_{S_i\in \calS}w'_i \cdot \fun(S_i,F)\in
(1\pm \e)\sum_{S_i\in \boldS}w_i\cdot \fun(S_i,F)
$$
(or more compactly, $\fun(\calS, F) \in(1\pm\e) \fun(\boldS, F)$
\footnote{
	The notation $(1\pm \e)B$ means
	the interval $[(1-\epsilon)B, (1+\epsilon)B].$
	}
).
We denote the cardinality
of the coreset $\calS$ as $|\calS|$.
\end{definition}

Definition~\ref{def:core} also generalizes the prior definition in~\cite{varadarajan2012sensitivity}, where each $S_i\in \boldS$
contains only one point.

\paragraph{Total sensitivity and dimension}
To bound the size of the generalized
coresets, we need the notion of
{\em total sensitivity}, originally introduced in~\cite{langberg2010}.

\begin{definition} (Total sensitivity of a
	generalized shape fitting instance).
\label{def:totalsen}
Let $\boldR^d$ be the collection of all finite discrete point sets $P\subset \R^d$, and let $\fun:\boldR^d\times \calF\rightarrow \R^{\geq 0}$ be a continuous function. Given an instance $\boldS=\{S_i\mid S_i\subset \boldR^d, 1\leq i\leq n\}$ of a generalized shape fitting problem $(\R^d,\calF,\fun)$,
with a weight function $w:\boldS\rightarrow \R^+$, the sensitivity $S_i\in \boldS$ is $\sigma_{\boldS}(S_i):=\inf\{\beta\geq 0\mid w_i\cdot \fun(S_i,F)\leq \beta \cdot \fun(\boldS,F), \forall F\in \calF\}$.
The total sensitivity of $\boldS$ is defined by $\mathfrak{G}_{\boldS}=\sum_{S_i\in \boldS}\sigma_{\boldS}(S_i)$.
\end{definition}

Note that this definition generalizes the one in~\cite{langberg2010}. In fact, if each $S_i\in \boldS$ contains only one point and the weight function $w_i=1$ for all $i$, this definition is equivalent to the definition in~\cite{langberg2010}.

We also need to generalize the definition of {\em dimension}
defined in \cite{FL11} (it is in fact
the primal shattering dimension
(See e.g., \cite{FL11, har2011geometric}) of a certain range space.
It plays a similar role to VC-dimension).

\begin{definition}
\label{def:dim} (Generalized dimension)
\eat{
Let $P\in \R^d$ be an instance of a shape fitting problem $(\R^d,\calF,\dist)$. Given a weight function $w:P\rightarrow \R^+$, consider the set system $(P,\calR)$, where $\calR$ is a family of subsets $R_{F,r}$ of $P$ defined as follows: given an $F\in \calF$ and $r\geq 0$, let $R_{F,r}=\{p\in P\mid w(p)\cdot \dist(p,F)\geq r\}\in \calR$ be the set of those points in $P$ whose weighted distance to the shape $F$ is at least $r$. We denote the dimension of the instance $P$ of the shape fitting problem by $\dim(P)$, to be the smallest integer $m$, such that for any weight function $w$ and $A\subseteq P$ of size $|A|=a\geq 2$, we have $|\{A\cap R_{F,r}\mid F\in \calF,r\geq 0\}|\leq a^m$.

In this paper, we generalize the above definition from~\cite{langberg2010}.
}
Let $\boldS=\{S_i\mid S_i\in \boldR^d, 1\leq i\leq n\}$ be an instance of
a generalized shape fitting problem $(\R^d,\calF,\fun)$.
Suppose $w_i$ is the weight of $S_i$.
We consider the range space $(\boldS,\calR)$, where $\calR$ is a family of subsets $R_{F,r}$ of $\boldS$ defined as follows: given an $F\in \calF$ and $r\geq 0$, let $R_{F,r}=\{S_i\in \boldS\mid w_i\cdot \fun(S_i,F)\geq r\}\in \calR$ consist of the sets $S_i$ whose weighted distance to the shape $F$ is at least $r$. Finally, we denote the generalized dimension of the instance $\boldS$ by $\dim(\boldS)$, to be the smallest integer $m$, such that for any weight function $w$ and $\calA\subseteq \boldS$ of size $|\calA|=a\geq 2$, we have $|\{\calA\cap R_{F,r}\mid F\in \calF,r\geq 0\}|\leq a^m$.
\end{definition}

The definition~\cite{langberg2010} is a special case of the above
definition when each $S_i\in \boldS$ contains only one point. On the other hand, the above definition is a special case of Definition 7.2 \cite{FL11} if thinking each $w_i\cdot \fun(S_i,\cdot )=g_i(\cdot)$ as a function from $\calF$ to $\R^{\geq 0}$.

We have the following lemma for bounding
the size of generalized
coresets by the generalized total sensitivity and dimension.
The proof is a straightforward extension of a result
in \cite{FL11}.
See Appendix~\ref{app:pre} for the details.

\begin{lemma}
\label{lm:sentocore}
Given any instance $\boldS=\{S_i\mid S_i\subset \boldR^d, 1\leq i\leq n\}$ of a generalized shape fitting problem $(\R^d,\calF,\fun)$, any weight function $w:\boldS\rightarrow \R^+$, and any $\e\in (0,1]$, there exists a generalized $\e$-coreset for $\boldS$ of
cardinality $O((\frac{\mathfrak{G}_{\boldS}}{\e})^2 \dim(\boldS))$.
\end{lemma}

\section{Stochastic Minimum $k$-Center}
\label{sec:exist}

In this section, we consider the stochastic minimum $k$-center problem
in $\R^d$ in the stochastic model.
Let $\calF$ be the family of all $k$-point sets of $\R^d$,
and let $\calP$ be the set of stochastic points.
Our main technique is to construct an \kcentercoreset\ $\calS$ of constant size.
For any $k$-point set $F\in \calF$, $\maxdist(\calS, F)$ should be
a ($1\pm \e$)-estimation for $\maxdist(\calP,F)=\Exp_{P\sim \calP}[\maxdist(P,F)]$. Recall that $\maxdist(P,F)=\max_{s\in P}\min_{f\in F}\dist(s,f)$ is the $k$-center value between two point sets $P$ and $F$.
Constructing $\calS$ includes two main steps: 1) Partition all realizations via additive $\e$-coresets, which reduces an exponential number of realizations to a polynomial number of
point sets.
2) Show that there exists a generalized coreset of constant cardinality
for the generalized $k$-median problem defined over
the above set of polynomial point sets.
Finally, we enumerate polynomially many possible collections $\calS_i$ (together with their weights).
We show that there is an \kcentercoreset\ $\calS$ among those candidate. By solving a polynomial system for each  $\calS_i$, and take the minimum solution, we can obtain a PTAS.

We first need the formal definition of an additive $\e$-coreset
\cite{agarwal2002exact} as follows.

\begin{definition} (additive $\e$-coreset)
\label{def:addcore} Let $B(f,r)$  denote the ball of radius $r$ centered at point $f$. For a set of points $P\in \boldR^d$, we call $Q\subseteq P$ an \emph{additive $\e$-coreset} of $P$ if for every $k$-point set $F=\{f_1,\ldots,f_k\}$, we have
$$ P\subseteq \cup_{i=1}^k B(f_i,(1+\e) \maxdist(Q,F)),
$$
i.e., the union of all balls $B(f_i,(1+\e) \maxdist(Q,F))$ $(1\leq i\leq k)$ covers $P$.
\footnote{Our definition is slight weaker than that in~\cite{agarwal2002exact}.
	The weaker definition suffices for our purpose.}
\end{definition}

\eat{
\begin{lemma} (\cite{agarwal2002exact})
\label{lm:addcorekcenter}
For a set of points $P\in \R^d$, there exists an \emph{additive $\e$-coreset} $Q\subseteq P$ of size $O(k/\e^d)$, which can be constructed in time $O(n+k/\e^d)$.
\end{lemma}
}

\subsection{Existential uncertainty model}

We first consider the existential uncertainty model.

\topic{Step 1: Partitioning realizations}

We first provide an algorithm \ACORESET,
which can construct an additive $\e$-coreset for any deterministic point set.
We can think \ACORESET\ as a mapping from all realizations of $\calP$
to all possible additive $\e$-coresets.
The mapping naturally induces a partition of all realizations.
Note that we do not run \ACORESET\ on every realization.

\topic{Algorithm \ACORESET\ for constructing additive $\e$-coresets.}
Given a realization $P\sim \calP$, we build a Cartesian grid $G(P)$ of side length depending on $P$. Let $\calC(P)=\{C\mid C\in G, C\cap P\neq \emptyset\}$ be the collection of those nonempty cells (i.e., cells that contain at least one point in $P$). In each non-empty cell $C\in \calC(P)$, we maintain the point $s^C\in C\cap P$ of smallest index. Let $\alg(P)=\{s^C\mid C\in G\}$, which
is an additive $\e$-coreset of $P$.
Finally the output of \ACORESET$(P)$ is $\alg(P),G(P)$, and $\calC(P)$.
The details can be found in Appendix~\ref{app:kcenter}.

Note that we do not use the construction of additive $\e$-coresets~\cite{agarwal2002exact}, because we
need the set of additive $\e$-coreset to have some extra properties
(in particular, Lemma~\ref{ob:prob} below), which allows us to compute certain probability values efficiently.

We first have the following lemma.

\begin{lemma}
\label{ob:2}
The running time of \ACORESET\ on any $n$ point set $P$ is $O(kn^{k+1})$. Moreover, the output $\alg(P)$ is an additive $\e$-coreset of $P$ of size at most $O(k/\e^d)$.
\end{lemma}

Denote $\alg(\calP)=\{\alg(P)\mid P\sim \calP\}$ be the collection of all possible additive $\e$-coresets. By Lemma~\ref{ob:2}, we know that each $S\in \alg(\calP)$ is of size at most $O(k/\e^d)$. Thus, the cardinality of $\alg(\calP)$
is at most $n^{O(k/\e^d)}$. For a point set $S$, denote $\Prob_{P\sim \calP}[\alg(P)=S]=\sum_{P:P\sim \calP,\alg(P)=S}\Prob[\vDash P]$ to be the
probability that the additive $\e$-coreset of a realization is $S$.
The following simple lemma states that we can have
a polynomial size representation for the objective function $\maxdist(\calP,F)$.

\begin{lemma}
\label{lm:step1}
Given $\calP$ of $n$ points in $\R^d$ in the existential uncertainty model, for any $k$-point set $F\in \calF$, we have that
$$
\sum_{S\in \alg(\calP)} \Prob_{P\sim \calP}[\alg(P)=S]\cdot \maxdist(S,F)\in (1\pm \e)\maxdist(\calP,F).
$$
\end{lemma}

\begin{proof}
By the definition of $\Prob_{P\sim \calP}[\alg(P)=S]$,
we can see that for any $k$-point set $F\in \calF$,
\begin{align*}
&\sum_{S\in \alg(\calP)} \Prob_{P\sim \calP}[\alg(P)=S]\cdot \maxdist(S,F)=\sum_{S\in \alg(\calP)} \sum_{P:P\sim \calP,\alg(P)=S}\Prob[\vDash P]\cdot  \maxdist(S,F)\\
\in &(1\pm \e)\sum_{S\in \alg(\calP)}\sum_{P:P\sim \calP,\alg(P)=S}\Prob[\vDash P]\cdot  \maxdist(P,F)=(1\pm \e)\maxdist(\calP,F).
\end{align*}
The inequality above uses the definition of additive $\e$-coresets (Definition~\ref{def:addcore}).
\end{proof}

We can think $\calP\rightarrow \alg(\calP)$ as a mapping, which maps a realization $P\sim \calP$ to its additive $\e$-coreset $\alg(P)$. The mapping partitions all realizations $P\sim \calP$ into a polynomial number of
additive $\e$-coresets.
For each possible additive $\e$-coreset $S\in \alg(\calP)$, we denote $\alg^{-1}(S)=\{P\sim \calP \mid \alg(P)=S\}$ to be the collection of all realizations mapping to $S$. By the definition of $\alg(\calP)$,
we have that $\cup_{S\in \alg(\calP)} \alg^{-1}(S)=\calP$.

Now, we need an efficient algorithm to compute
$\Prob_{P\sim \calP}[\alg(P)=S]$
for each additive $\e$-coreset $S\in \alg(\calP)$.
The following lemma states
that the mapping constructed by algorithm \ACORESET\ has
some nice properties that allow us to compute the probabilities.
This is also the reason why we cannot directly use the original
additive $\e$-coreset construction algorithm in \cite{agarwal2002exact}.
The proof is somewhat subtle and can be found in Appendix~\ref{app:kcenter}.

\begin{lemma}
\label{ob:prob}
Consider a subset $S$ of at most $O(k/\e^d)$ points.
Run algorithm \ACORESET$(S)$, which outputs
an additive $\e$-coreset $\alg(S)$, a Cartesian grid $G(S)$, and
a collection $\calC(S)$ of nonempty cells.
If $\alg(S)\neq S$, then $S\notin \alg(\calP)$
(i.e., $S$ is not the output of \ACORESET\ for any realization $P\sim \calP$).
\footnote{
	It is possible that some point set $S$ satisfies
	Definition~\ref{def:addcore} for some realization $P$, but
	is not the output of \ACORESET$(S)$.
	}
If $|S|\leq k$, then $\alg^{-1}(S)=\{S\}$.
Otherwise if $\alg(S)=S$ and $|S|\geq k+1$,
then a point set $P\sim \calP$ satisfies $\alg(P)=S$ if and only if
\begin{enumerate}
\item[P1.] For any cell $C\notin \calC(S)$, $C\cap P=\emptyset$.
\item[P2.] For any cell $C\in \calC(S)$, assume that point $s^C=C\cap S$.
Then $s^C\in P$, and any point $s'\in C\cap \calP$ with a smaller index
than that of $s^C$ does not appear in the realization $P$.
\end{enumerate}
\end{lemma}

Thanks to Lemma~\ref{ob:prob},
now we are ready to show how to compute $\Prob_{P\sim \calP}[\alg(P)=S]$ efficiently for each $S\in \alg(\calP)$.
We enumerate every point set of size $O(k/\e^d)$.
For a set $S$, we first run \ACORESET$(S)$ and output a Cartesian grid $G(S)$ and a point set $\alg(S)$.
We check whether $S\in \alg(\calP)$
by checking whether $\alg(S)=S$ or $|S|\leq k$.
If $S\in \alg(\calP)$, we can compute $\Prob_{P\sim \calP}[\alg(P)=S]$ using the Cartesian grid $G(S)$. See Algorithm~\ref{alg:prob} for details.

\begin{algorithm}
\caption{Computing $\Prob_{P\sim \calP}[\alg(P)=S\mathrm{]}$}
\label{alg:prob}

For each point set $S\sim \calP$ of size $|S|=O(k/\e^d)$, run algorithm \ACORESET$(S)$. Assume that the output is a point set $\alg(S)$, a Cartesian grid $G(S)$, and a cell collection $\calC(S)=\{C\mid C\in G, C\cap S\neq \emptyset\}$. \\
If $\alg(S)\neq S$, output $\Prob_{P\sim \calP}[\alg(P)=S]=0$. If $|S|\leq k$, output $\Prob_{P\sim \calP}[\alg(P)=S]=\Prob[\vDash S]$.\\
For a cell $C\notin \calC(S)$, suppose $C\cap \calP=\{t_i\mid t_i\in \calP, 1\leq i\leq m\}$. W.l.o.g., assume that $t_1,\ldots, t_m$ are in increasing order of their indices.
For $C\not\in \calC(S)$, let
$$
Q(C)=\Prob_{P\sim \calP}\Bigl[P\cap C=\emptyset\Bigr]=\prod_{i=1}^{m}(1-p_i)
$$
be the probability that no point in $C$ is realized.
If $C\in \calC(S)$, assume that point $t_j\in C\cap S$, and
let
$$Q(C)=\Pr_{P\sim \calP}\Bigl[t_j\in P\text{ and }
\{t_1,\ldots, t_{j-1}\}\cap P=\emptyset\Bigr]=p_j\cdot \prod_{i=1}^{j-1}(1-p_i)$$
be the probability that $t_j$ appears, but $t_1,\ldots, t_{j-1}$ do not appear. \\
Output $\Prob_{P\sim \calP}[\alg(P)=S]=\prod_{C\in G(S)}Q(C)$.

\end{algorithm}

The following lemma asserting the correctness of Algorithm~\ref{alg:prob}
is a simple consequence of Lemma~\ref{ob:prob}.

\begin{lemma}
\label{lm:correct}
For any point set $S$, Algorithm~\ref{alg:prob} computes exactly the total probability $$\Prob_{P\sim \calP}[\alg(P)=S]=\sum_{P:P\sim \calP,\alg(P)=S}\Prob[\vDash P]$$ in $O(n^{O(k/\e^d)})$ time.
\end{lemma}
\begin{proof}
Run \ACORESET$(S)$, and we obtain a point set $\alg(S)$. If $\alg(S)\neq S$, we have that $S\notin \alg(\calP)$ by Lemma~\ref{ob:prob}. Thus, $\Prob_{P\sim \calP}[\alg(P)=S]=0$.
If $|S|\leq k$, we have that $\alg^{-1}(S)=\{S\}$ by Lemma~\ref{ob:prob}. Thus, $\Prob_{P\sim \calP}[\alg(P)=S]=\Prob[\vDash S]$.

Otherwise if $\alg(S)=S$ and $|S|\geq k+1$, by Lemma~\ref{ob:prob}, each realization $P\in \alg^{-1}(S)$ satisfies P1 and P2. Then combining the definition of $Q(C)$,
and the independence of all cells, we can see that $\prod_{C\in \calC}Q(C)$ is equal to $\sum_{P\in \alg^{-1}(S)}\Prob[\vDash P]=\Prob_{P\sim \calP}[\alg(P)=S]$.

For the running time, note that we only need to consider at most $n^{O(k/\e^d)}$ point sets $S\sim \calP$. For each $S$, Algorithm~\ref{alg:prob} needs to run \ACORESET$(S)$, which costs $O(kn^{k+1})$ time by Lemma~\ref{ob:2}. Step 2 and 3 only cost linear time. Thus, we can compute all probabilities $\Prob_{P\sim \calP}[\alg(P)=S]$ in $O(n^{O(k/\e^d)})$ time.
\end{proof}

\topic{Step 2: Existence of generalized coreset via generalized total sensitivity}

Recall that $\alg(\calP)$ is a collection of polynomially many point sets of size $O(k/\e^d)$. By Lemma~\ref{lm:step1},
we can focus on a generalized $k$-median problem: finding a $k$-point set $F\in \calF$ which minimizes $\maxdist(\alg(\calP),F)=\sum_{S\in \alg(\calP)} \Prob_{P\sim \calP}[\alg(P)=S]\cdot \maxdist(S,F)$.
In fact, the generalized $k$-median problem is a special case
of the generalized shape fitting problem we defined in Definition~\ref{def:shape}.
Here, we instantiate the shape family $\calF$ to be the collection of all $k$-point sets.
Note that the $k$-center objective $\maxdist(\alg(\calP),F)$ is indeed a 
generalized distance function in Definition~\ref{def:shape}.
To make things concrete, we formalize it below.
Recall that $\boldR^d$
is the collection of all finite discrete point sets in $\R^d$.

\begin{definition}
\label{def:kcore}
A generalized $k$-median problem is specified by a triple $(\R^d, \calF, \maxdist)$. Here $\calF$ is the family of all $k$-point sets in $\R^d$, and $\maxdist: \boldR^d\times \calF\rightarrow \R^{\geq 0}$ is a generalized distance function defined as follows: for a point set $P\in \boldR^d$ and a $k$-point set $F\in \calF$, $\maxdist(P,F)=\max_{s\in P}\dist(s,F)=\max_{s\in P}\min_{f\in F}\dist(s,f)$. An instance $\boldS$ of the generalized $k$-median problem is a (weighted) collection $\{S_1, \ldots, S_m\}$ ($S_i\in \boldR^d$)
of point sets, and each $S_i$ has a positive weight $w_i\in \R^+$. For any $k$-point set $F\in \calF$, the total generalized distance from $\boldS$ to $F$ is $\setmaxdist(\boldS,F)=\sum_{S_i\in \boldS}w_i\cdot \maxdist(S_i,F)$.
The goal of the generalized $k$-median problem (GKM) is to find a $k$-point set $F$ which minimizes the total generalized distance $\setmaxdist(\boldS,F)$.
\end{definition}

Recall that a generalized $\e$-coreset
is a sub-collection $\calS\subseteq \boldS$ of point sets,
together with a weight function $w':\calS\rightarrow \R^+$,
such that for any $k$-point set $F\in \calF$, we have
$
\sum_{S\in \calS}w'(S)\cdot \maxdist(S,F)\in (1\pm \e)\sum_{S\in \boldS}w(S)\cdot \maxdist(S,F)
$
(or $\maxdist(\calS, F) \in(1\pm\e) \maxdist(\boldS, F)$).
This generalized coreset will serve as
the \kcentercoreset\ for the original stochastic $k$-center problem.

Our main lemma asserts that
a constant sized generalized coreset exists, as follows.

\begin{lemma} (main lemma)
\label{lm:constant}
Given an instance $\calP$ of $n$ stochastic points in $\R^d$, let $\alg(\calP)$ be the collection of all additive $\e$-coresets. There exists a generalized $\e$-coreset $\calS\subseteq \alg(\calP)$ of cardinality $|\calS|=O(\e^{-(d+2)}dk^4)$, together with a weight function $w':\calS \rightarrow \R^+$, which satisfies that for any $k$-point set $F\in \calF$,
$$ \sum_{S\in \calS}w'(S)\cdot\maxdist(S,F) \in (1\pm \e)\sum_{S\in \alg(\calP)} \Prob_{P\sim \calP}[\alg(P)=S]\cdot \maxdist(S,F).
$$
\end{lemma}

\eat{
In this paper, we consider generalized definitions for Definition~\ref{def:totalsen} and~\ref{def:dim}. We allow an instance $P\in \boldR^d$ to be a collection $P=\{S_i\mid S_i\in \boldR^d, 1\leq i\leq n\}$, i.e., an element $p\in P$ becomes a subset in $\R^d$. Note that this change generalizes the two definitions, since that if each $S_i\in \R^d$ is a single point, then $P\in \boldR^d$. By this generalization, the weighted $k$-median problem in Definition~\ref{def:kcore} can also be considered as a generalized shape fitting problem $(\R,\calF,\maxdist)$.
}

Now, we prove Lemma~\ref{lm:constant} by showing a constant upper bound on
the cardinality of a generalized $\e$-coreset.
This is done by applying Lemma~\ref{lm:sentocore} and providing
constant upper bounds for both the total sensitivity and the generalized dimension of the
generalized $k$-median instance.

Given an instance $\boldS=\{S_i\mid S_i\in \boldR^d, 1\leq i\leq n\}$ of a generalized $k$-median problem with a weight function $w:\boldS\rightarrow \R^+$,
we denote $F^*$ to be the $k$-point set which minimizes the total generalized distance $\setmaxdist(\boldS,F)=\sum_{S\in \boldS}w(S)\cdot\maxdist(S,F)$
over all $F\in \calF$. W.l.o.g., we assume that $\setmaxdist(\boldS,F^*)>0$.
Since if $\setmaxdist(\boldS,F^*)=0$,
there are at most $k$ different points in the instance.

We first construct a {\em projection instance} $P^*$ of a weighted $k$-median problem for $\boldS$, and relate the total sensitivity $\mathfrak{G}_{\boldS}$ to $\mathfrak{G}_{P^*}$.
Recall that $\mathfrak{G}_{\boldS}=\sum_{S\in \boldS}\sigma_{\boldS}(S)$ is the total sensitivity of $\boldS$. Our construction of $P^*$ is as follows. For each point set $S_i\in \boldS$, assume that $F^*_i\in \calF$ is the $k$-point set satisfying that $F^*_i=\mathsf{argmax}_{F}\frac{w(S_i)\cdot \maxdist(S_i,F)}{\setmaxdist(\boldS,F)}$, i.e., the sensitivity $\sigma_{\boldS}(S_i)$ of $S_i$ is equal to $\frac{w(S_i) \maxdist(S_i,F^*_i)}{\setmaxdist(\boldS,F^*_i)}$.
Let $s^*_i\in S_i$ denote the point farthest to $F^*_i$.
\footnote{If more than 1 points in $S_i$ have this property, we arbitrarily choose one.}
Let $f^*_i\in F^*$ denote the point closest to $s^*_i$.
Denote $P^*$ to be the multi-set $\{f^*_i\mid S_i\in \boldS\}$, and denote the weight function $w':P^*\rightarrow \R^+$ to be $w'(f^*_i)=w(S_i)$ for any $i\in [n]$. Thus, $P^*$ is a weighted $k$-median instance in $\R^d$ with a weight function $w'$. See Figure~\ref{fig:P} for an example of the construction of $P^*$.

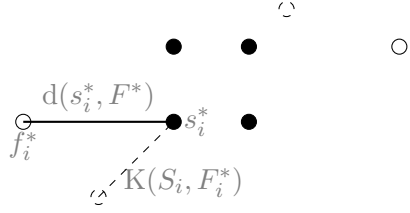
\begin{figure}[t]
  \centering
  \begin{tikzpicture}

    \node[gray, below ] at (-2,0) {$f^*_i$} ;
    \node[gray, above ] at (-1,0) {$\dist(s^*_i, F^*)$};
    \node[gray, above,right ] at (0,0) {$s^*_i$} ;
    \node[gray, right] at (-0.8,-0.8) {$\maxdist(S_i,F^*_i)$};

    \draw[fill] (0,0) circle (0.1);
    \draw[fill] (1,0) circle (0.1)  ;
    \draw[fill] (0,1) circle (0.1)  ;
    \draw[fill] (1,1) circle (0.1) ;
    \draw[] (-2,0) circle (0.1) ;
    \draw[] (3,1) circle (0.1);
    \draw[dashed] (-1,-1) circle (0.1) ;
    \draw[dashed] (1.5,1.5) circle (0.1) ;
    \draw[thick] (-2,0) -- (0,0);
    \draw[dashed] (-1,-1) -- (0,0);
  \end{tikzpicture}

  \caption{ In the figure, $S_i$ is the black point set,
  	$F^*$ is the white point set, and $F^*_i$ is the dashed point set.
  	Here, $s^*_i\in S_i$ is the farthest point to $F^*_i$ satisfying $\dist(s^*_i,F^*_i)=\maxdist(S_i,F^*_i)$, and $f^*_i\in F^*$ is the closest point to $s^*_i$ satisfying $\dist(s^*_i,f^*_i)=\dist(s^*_i,F^*)$.}
  \label{fig:P}
\end{figure}

\begin{lemma}
\label{lm:dimreduction} Given an instance $\boldS=\{S_i\mid S_i\in \boldR^d, 1\leq i\leq n\}$ of a generalized $k$-median problem in $\R^d$ with a weight function $w:\boldS\rightarrow \R^+$, let $P^*$ be its projection instance. Then, we have $\mathfrak{G}_{\boldS}\leq 2\mathfrak{G}_{P^*}+1$.
\end{lemma}

\begin{proof}
First note that we have the following fact. Given $i,j\in [n]$, recall that $s^*_j\in S_j$ is the farthest point to $F^*_j$, and $f^*_j\in F^*$ is the closest point to $s^*_j$. Let $f\in F^*_i$ be the point closest to $s^*_j$.
\begin{align}
\nonumber
\maxdist(S_j,F^*_i)+\maxdist(S_j,F^*)& \geq \dist(s^*_j,F^*_i)+\dist(s^*_j,F^*)= \dist(s^*_j,F^*_i)+\dist(s^*_j,f^*_j) \\
&=\dist(s^*_j,f)+\dist(s^*_j,f^*_j)\geq \dist(f^*_j,f)\geq \dist(f^*_j,F^*_i),
\label{eq:key}
\end{align}
The first inequality follows from the definitions of $\maxdist(S_j,F^*_i)$ and $\maxdist(S_j,F^*)$. The first equality follows from the definition of $f^*_j$. The second inequality follows from the triangle inequality, and the last inequality is by the definition of $\dist(f^*_j,F^*_i)$.

Then we have the following fact:
\begin{align}
\label{ineq:2}
\sum_{f\in P^*}w'(f)\cdot\dist(f,F^*_i)&=\sum_{f^*_j\in P^*}w'(f^*_j)\cdot\dist(f^*_j,F^*_i)\leq \sum_{S_j\in \boldS} w(S_j)\cdot\bigl( \maxdist(S_j,F^*)+\maxdist(S_j,F^*_i)\bigr)\notag\\
&=\setmaxdist(\boldS,F^*)+\setmaxdist(\boldS,F^*_i)\leq 2\setmaxdist(\boldS,F^*_i),
\end{align}
since $\setmaxdist(\boldS,F^*)\leq \setmaxdist(\boldS,F^*_i)$ and Inequality (\ref{eq:key}).

Let $f'\in F^*_i$ be the point closest to $f^*_i$.
We also notice the following fact:
\begin{align}
\label{ineq:1}
\maxdist(S_i,F^*)+\dist(f^*_i,F^*_i)& \geq \dist(s^*_i,f^*_i)+\dist(f^*_i,F^*_i)=\dist(s^*_i,f^*_i)+\dist(f^*_i,f')\notag \\
&\geq \dist(s^*_i,f')\geq \dist(s^*_i,F^*_i)=\maxdist(S_i,F^*_i).
\end{align}
The first inequality follows from the definition of $f^*_i$, the second inequality follows from the triangle inequality, and the last inequality follows from the definition of $\dist(s^*_i,F^*_i)$.

Now we are ready to analyze $\sigma_{\boldS}(S_i)$ for some $S_i\in \boldS$.
We can see that
\begin{align*}
w(S_i)\cdot \maxdist(S_i,F^*_i)&
\leq w(S_i)\cdot\maxdist(S_i,F^*)+w(S_i)\cdot\dist(f^*_i,F^*_i) &\quad \text{ [by  \eqref{ineq:1}]} \\
& \leq
 w(S_i)\cdot\maxdist(S_i,F^*)+\sigma_{P^*}(f^*_i)\cdot
 \biggl(\sum_{f\in P^*}w'(f)\cdot\dist(f,F^*_i)\biggr)
 & \text{[by the definition of $\sigma_{P^*}$]}\\
&\leq
 w(S_i)\cdot\maxdist(S_i,F^*)+2\sigma_{P^*}(f^*_i)\cdot \setmaxdist(\boldS,F^*_i)
 & \text{[by \eqref{ineq:2}]}\\
&= \frac{w(S_i)\cdot\maxdist(S_i,F^*)}{\setmaxdist(\boldS,F^*_i)}\cdot \setmaxdist(\boldS,F^*_i)+ 2\sigma_{P^*}(f^*_i)\cdot \setmaxdist(\boldS,F^*_i)\\
&\leq \left(\frac{w(S_i)\cdot\maxdist(S_i,F^*)}{\setmaxdist(\boldS,F^*)}
+2\sigma_{P^*}(f^*_i)\right)\setmaxdist(\boldS,F^*_i). & \text{[by $\maxdist(\boldS,F^*_i)\geq \maxdist(\boldS,F^*)$]}
\end{align*}
Finally, we bound the total sensitivity as follows:
$$ \mathfrak{G}_{\boldS}=
\sum_{S_i\in \boldS}\sigma_{\boldS}(S_i)\leq
\sum_{S_i\in \boldS}
\left(\frac{w(S_i)\cdot\maxdist(S_i,F^*)}{\setmaxdist(\boldS,F^*)}+
2\sigma_{P^*}(f^*_i)\right)=1+2\mathfrak{G}_{P^*}.
$$
This finishes the proof of the lemma.
\end{proof}

Since $P^*$ is an instance of a weighted $k$-median problem,
we know that the total sensitivity $\mathfrak{G}_{P^*}$ is at most $2k+1$,
by \cite[Theorem 9]{langberg2010}.
\footnote{Theorem 9 in~\cite{langberg2010}
	bounds the total sensitivity for the unweighted version.
	However, the proof can be extended to
	the weighted version in a straightforward way.
	}
Then combining Lemma~\ref{lm:dimreduction}, we have the following lemma which bounds the total sensitivity of $\mathfrak{G}_{\boldS}$.

\begin{lemma}
\label{lm:totalsen2}
Consider an instance $\boldS$ of a generalized $k$-median problem $(\R^d,\calF,\maxdist)$. The total sensitivity $\mathfrak{G}_{\boldS}$ is at most $4k+3$.
\end{lemma}

Now the remaining task is to bound the generalized dimension $\dim(\boldS)$.
Consider the range space $(\boldS,\calR)$, $\calR$ is a family of subsets $R_{F,r}$ of $\boldS$ defined as follows: given an $F\in \calF$ and $r\geq 0$, let $R_{F,r}=\{S_i\in \boldS\mid w_i\cdot \maxdist(S_i,F)\geq r\}\in \calR$. Here $w_i$ is the weight of $S_i\in \boldS$. We have the following lemma.

\eat{
\begin{lemma}
\label{lm:quant1}
Let $\boldS_1=(\boldS,\calR^1),\ldots,\boldS_k=(\boldS,\calR^k)$ be set systems with generalized dimension $\delta_1,\ldots,\delta_k$, respectively. Consider the following set system
$$ \calR'=\{r_1\cup r_2\cup \cdots \cup r_k\mid r_1\in \calR^1,\ldots, r_k\in \calR^k\}
$$
and the associated set system $(\boldS,\calR')$. Then, the generalized dimension of $(\boldS,\calR')$ is bounded by $O(k\delta)$,
where $\delta=\max_i \delta_i$.
\end{lemma}

\begin{proof}
For any weight function $w:\boldS\rightarrow \R^+$ and $\calA\subseteq \boldS$ of size $|\calA|=a\geq 2$, we consider the number $|\{\calA\cap R'_{F,r}\mid F\in \calF, r\geq 0\}|$.
For some $F\in \calF$ and $r\geq 0$, observe that $R'_{F,r}=R^1_{F_1,r_1}\cup R^2_{F_2,r_2}\cup\cdots \cup R^k_{F_k,r_k}$ for some $F_i\in \calF$ and $r_i\geq 0$ ($1\leq i\leq k$). So we have
$$ \calA\cap R'_{F,r}=(\calA\cap R^1_{F_1,r_1})\cup \cdots \cup (\calA\cap R^k_{F_k,r_k})
$$
By Definition~\ref{def:dim}, we know that $|\{\calA\cap R^{i}_{F,r}\mid F\in \calF, r\geq 0\}|\leq a^{\delta_i}$ for all $1\leq i\leq k$. Thus, we have
$$ |\{\calA\cap R'_{F,r}\mid F\in \calF, r\geq 0\}|\leq a^{\delta_1}\times a^{\delta_2}\times \cdots \times a^{\delta_k}\leq a^{k\delta},
$$ which proves the lemma.
\end{proof}

The proof of Lemma~\ref{lm:quant1} is standard, and very similar to \cite[Theorem 5.22]{har2011geometric}. We are now ready to show the following lemma.
}

\begin{lemma}
\label{lm:dim}
Consider an instance $\boldS$ of a generalized $k$-median problem in $\R^d$. If each point set $S\in \boldS$ is of size at most $L$, then the generalized dimension $\dim(\boldS)$ is $O(dkL)$.
\end{lemma}

\begin{proof}
Consider a mapping $g:\boldS\rightarrow \R^{dL}$ constructed as follows:
suppose $S_i=\{x^1=(x^1_1,\ldots, x^1_d),\ldots, x^L=(x^L_1,\ldots, x^L_d)\}$
(if $|S_i|<L$, we pad it with $x^1=(x^1_1,\ldots, x^1_d)$).
We let
$$
g(S_i)=(x^1_1,\ldots, x^1_d,\ldots, x^L_1,\ldots, x^L_d)\in \R^{dL}.
$$
For any $t\geq 0$ and any $k$-point set $F\in \calF$,
we observe that $w_i\cdot \maxdist(S_i,F)\geq r$ holds
if and only if there exists some $1\leq j\leq L$ satisfying that
$w_i \cdot \dist(x^j,F) \geq r$,
which is equivalent to saying that
point $g(S_i)$ is in the union of the following $L$ sets $\{(x^1_1,\ldots, x^1_d,\ldots, x^L_1,\ldots, x^L_d) \mid \dist(x^j,F)\geq r/w_i\}$ ($j\in [L]$).

Let $X$ be the image set of $g$. Let $(X,\calR^{j})$ ($1\leq j\leq L$) be $L$ range spaces, where
each $\calR^j$ consists of all subsets $R^j_{F,r}=\{(x^1_1,\ldots, x^1_d,\ldots, x^L_1,\ldots, x^L_d)\in X \mid \dist(x^j,F)\geq r\}$ for all $F\in \calF$ and $r\geq 0$. Note that each $(X,\calR^{j})$ has VC-dimension $dk$ by \cite{FL11}.
Thus, we have that each $(X,\calR^{j})$ has shattering dimension at most its VC-dimension $dk$ by Corollary 5.12 in \cite{har2011geometric}. Let $\calR'=\{\cup R_j\mid R_j\in \calR^j, i\in [L]\}$.
Using the standard result for bounding the shattering dimension
of the union of set systems (e.g.,\cite[Thm 5.22]{har2011geometric}), we can see that the shattering dimension of $(X,\calR')$ (which is the 
generalized dimension of $\boldS$) is bounded by $O(dkL)$.
\end{proof}

Note that an additive $\e$-coreset is of size at most $O(k/\e^d)$. Then combining Lemma~\ref{lm:sentocore},~\ref{lm:totalsen2} and~\ref{lm:dim}, we directly obtain Lemma~\ref{lm:constant}.
Combining Lemma~\ref{lm:step1} and~\ref{lm:constant}, we have the following theorem.

\begin{theorem}
\label{thm:existpras}
Given an instance $\calP$ of $n$ points in $\R^d$ in the existential uncertainty model, there exists an \kcentercoreset\ $\calS$ of $O(\e^{-(d+2)}dk^4)$ point sets with a weight function $w':\calS\rightarrow \R^+$, which satisfies that,
\begin{enumerate}
\item For each point set $S\in \calS$, we have $S\subseteq \calP$ and $|S|=O(k/\e^d)$.
\item For any $k$-point set $F\in \calF$, we have $\sum_{S\in \calS}w'(S)\cdot \maxdist(S,F)\in (1\pm \e)\maxdist(\calP,F).$
\end{enumerate}
\end{theorem}

\topic{PTAS for stochastic minimum $k$-center.} It remains to give a PTAS for the stochastic minimum $k$-center problem. For an instance $\alg(\calP)$ of a generalized $k$-median problem, if we can compute the sensitivity $\sigma_{\alg(\calP)}(S)$ efficiently for each point set $S\in \alg(\calP)$, then we can construct an \kcentercoreset\ by importance sampling
(The details of importance sampling can be found in Theorem 4.9 in~\cite{anthony2009neural}).
However, it is unclear how to compute the sensitivity $\sigma_{\alg(\calP)}(S)$ efficiently. Instead, we enumerate all weighted sub-collections $\calS_i\subseteq \alg(\calP)$ of cardinality at most $O(\e^{-(d+2)}dk^4)$. We claim that we only need to enumerate $O(n^{O(\e^{-(2d+2)}dk^5)})$ polynomially many sub-collections $\calS_i$ together with their weight functions, such that there exists a generalized $\e$-coreset of $\alg(\calP)$.
\footnote{We remark that even though we enumerate the weight function,
		computing $\Prob_{P\sim \calP}[\alg(P)=S]$ is still important for our algorithm. See Lemma~\ref{lm:coreext} for the details of the
		enumeration algorithm.	}
The details can be found in Appendix~\ref{app:core}.

In the next step, for each weighted sub-collection $\calS\subseteq \alg(\calP)$ with a weight function $w':\calS\rightarrow \R^+$, we briefly sketch how to compute the optimal $k$-point set $F$ such that $\maxdist(\calS,F)$ is minimized. 
We cast the optimization problem as a constant size polynomial system. 

Denote the space $\calF=\{(y^1,\ldots,y^k)\mid y^i\in \R^d, 1\leq i\leq k\}$ to be the collection of ordered $k$-point sets 
($(y^1,y^2,\ldots,y^k)\in \calF$ and $(y^2,y^1,\ldots,y^k)\in \calF$ to be two different $k$-point sets if $y^1\neq y^2$). 
We first divide the space $\calF$ into pieces $\{\calF^i\}$,
as follows:
Let $L=O(k/\e^d)$ and $\calL= (l_1,\ldots,l_L)$ $(1\leq l_j\leq k,\forall j\in [L])$ be a sequence of integers, and let $b\in [L]$ be an index.
Consider a point set $S=\{x^1=(x^1_1,\ldots, x^1_d),\ldots, x^L=(x^L_1,\ldots, x^L_d)\}\in \calS$ and a $k$-point set $F=\{y^1=(y^1_1,\ldots, y^1_d),\ldots, y^k=(y^k_1,\ldots, y^k_d)\}\in \calF$. We give the following definition. 
\begin{definition}
\label{def:decide}
The $k$-center value $\maxdist(S,F)$ is \emph{decided} by $\calL$ and $b$ if the following two properties hold.
\begin{enumerate}
\item For any $i\in [L]$ and any $j\in [k]$, $\dist(x^i,y^{l_i})\leq \dist(x^i,y^j)$, i.e., the closest point to $x^j$ is $y^{l_j}\in F$.
\item For any $i\in [L]$, $\dist(x^i,y^{l_i})\leq \dist(x^{b},y^{l_{b}})$, i.e., the $k$-center value $\maxdist(S,F)=\dist(x^{b},y^{l_{b}})$.
\end{enumerate}
\end{definition}
For each point set $S_i\in \calS$, we enumerate an integer sequence $\calL_i$ and an index $b_i$.
Given a collection $\{\calL_i,b_i\}_i$ (index $i$ ranges over all $S_i$ in $\calS$), we construct a piece $\calF^{\{\calL_i,b_i\}_i}\subseteq \calF$ as follows: for any point set $S_i\in \calS$ and any $k$-point set $F\in \calF^{\{\calL_i,b_i\}_i}$, the $k$-center value $\maxdist(S_i,F)$ is \emph{decided} by $\calL_i$ and $b_i$. 
According to Definition~\ref{def:decide}, $\calF^{\{\calL_i,b_i\}_i}$ is defined by a polynomial system.

Then, we solve our optimization problem in each piece $\calF^{\{\calL_i,b_i\}_i}$. 
By definition~\ref{def:decide}, for any point set $S_i\in \calS$ and any $k$-point set $F\in \calF^{\{\calL_i,b_i\}_i}$, the $k$-center value $\maxdist(S_i,F)=\dist(x^{b_i},y^{\calL_i(b_i)})$ ($x^{b_i}\in S_i$, $y^{\calL_i(b_i)}\in F$). Here, the index $\calL_i(b_i)$ is the $b_i$-th item of $\calL_i$. Hence, our problem can be formulated as the following optimization problem:
$$
\min_{F}\sum_{S_i\in \calS}w'(S_i)\cdot g_i, \quad \text{s.t.},\,\, g^2_i=\|x^{b_i}-y^{\calL_i(b_i)}\|^2, g_i\geq 0, \forall i\in [L]; y^{\calL_i(b_i)}\in F; F\in \calF^{\{\calL_i,b_i\}_i}.
$$
By Definition \ref{def:decide}, there are at most $kL|\calS|$ constraints, which is a constant. Thus, the polynomial system has $dk$ variables and $O(kL|\calS|)$ constraints, hence can be solved in constant time. Note that there are at most $O(k^{L|\calS|})$ different pieces $\calF^{\{\calL_i,b_i\}_i}\subseteq \calF$, which is again a constant. Thus, we can compute the optimal $k$-point set for the weighted sub-collection $\calS$ in constant time.

Now we return to the stochastic minimum $k$-center problem. Recall that we first enumerate all possible weighted sub-collections $\calS_i\subseteq \alg(\calP)$ of cardinality at most $O(\e^{-(d+2)}dk^4)$. Then we compute the optimal $k$-point set $F^i$ for each weighted sub-collection $\calS_i$ as above, and compute the expected $k$-center value $\maxdist(\calP,F^i)$.
\footnote{It is not hard to compute $\maxdist(\calP,F^i)$ in $O(n\log n)$ time by sorting all points in $\calP$ in non-increasing order according to their distances to $F^i$.}
Let $F^*\in \calF$ be the $k$-point set which minimizes the expected $k$-center value $\maxdist(\calP,F^i)$ over all $F^i$. By Lemma~\ref{lm:coreext}, there is one sub-collection $\calS_i$ with a weight function $w'$ satisfying that $\maxdist(\calS_i,F^i)\leq (1+ \e)\min_{F\in \calF} \maxdist(\calP,F)$. Thus, we conclude that $F^*$ is a $(1+\e)$-approximation for the stochastic minimum $k$-center problem. For the running time, we enumerate at most $O(n^{O(\e^{-(2d+2)}dk^5)})$ weighted sub-collections. Moreover, computing the optimal $k$-point set for each sub-collection costs constant time. Then the total running time is at most $O(n^{O(\e^{-(2d+2)}dk^5)})$. Thus, we have the following corollary.

\begin{corollary}
\label{cor:pras}
If both $k$ and $d$ are constants, given an instance $\calP$ of $n$ stochastic points in $\R^d$ in the existential uncertainty model, there exists a PTAS for the stochastic minimum $k$-center problem in $O(n^{O(\e^{-(2d+2)}dk^5)})$ time.
\end{corollary}

\subsection{Locational uncertainty model}
\label{sec:location}

Next, we consider the stochastic minimum $k$-center problem in the locational uncertainty model. Given an instance of $n$ nodes $u_1,\ldots,u_n$ which may locate in the point set $\calP= \{s_1,\ldots,s_m\mid s_i\in \R^d, 1\leq i\leq m\}$, our construction of additive $\e$-coresets and the method for bounding the total sensitivity is exactly the same as in the existential uncertainty model. The only difference is that for an additive $\e$-coreset $S$, how to compute the probability $\Prob_{P\sim \calP}[\alg(P)=S]=\sum_{P:P\sim \calP,\alg(P)=S }\Prob[\vDash P]$. Here, $P\sim \calP$ is a realized point set according to the probability distribution of $\calP$. Run \ACORESET$(S)$, and construct a Cartesian grid $G(S)$. 
Denote $T(S)=\bigl(\cup_{P:P\sim \calP, \alg(P)=S} P\bigr) \setminus S$ to be the collection of all points $s$ which might be contained in some realization $P\sim \calP$ with $\alg(P)=S$. Recall that $\calC(S)=\{C\in G\mid |C\cap S|=1\}$ is the collection of $d$-dimensional Cartesian cells $C$ which contains a point $s^C\in S$. By Lemma \ref{ob:prob}, for any realization $P$ with $\alg(P)=S$, we have the following observations.
\begin{enumerate}
\item For any cell $C\notin \calC(S)$, $C\cap P=\emptyset$. It means that for any point $s\in C\cap \calP$, we have $s\notin T(S)$.
\item For any cell $C\in \calC(S)$ and any point $s'\in C\cap \calP$ with a smaller index than that of $s^C$, we have $s'\notin P$. It means that $s'\notin T(S)$. 
\end{enumerate}
By the above observations, we conclude that $T(S)$ is the collection of those points $s'$ belonging to some cell $C\in \calC(S)$ and with a larger index than that of $s^C$.

Then we reduce the counting problem $\Prob_{P\sim \calP}[\alg(P)=S]$ to a family of bipartite holant problems. We first give the definition of holant problems.

\begin{definition}
An instance of a holant problem is a tuple
$\Lambda=\tuple{G(V,E),\tuple{g_v}_{v\in V}},\tuple{w_e}_{e\in E}$, where for every
$v\in V$, $g_v:\set{0,1}^{E_v}\to \mathbb{R}^+$ is a function, where $E_v$ is the set of edges incident to $v$. For every assignment
$\sigma\in \set{0,1}^{E}$, we define the weight of $\sigma$
as
\[
  w_{\Lambda}(\sigma)\triangleq \prod_{v\in
    V}g_v\tuple{\sigma\mid_{E_v}}\prod_{e\in \sigma}w_e.
\]
Here $\sigma\mid_{E_v}$ is the assignment of $E_v$ with respect to the assignment $\sigma$. We denote the value of the holant problem $Z(\Lambda)\triangleq\sum_{\sigma\in \set{0,1}^{E}}w_\Lambda(\sigma).$
\end{definition}

For a counting problem $\Prob_{P\sim \calP}[\alg(P)=S]$, w.l.o.g., we assume that $S=\{s_1,\ldots, s_{|S|}\}$. Then we construct a family of holant instance $\Lambda_{\calL}$ as follows.

\begin{enumerate}
\item Enumerate all integer sequences $\calL=(l_1,\ldots,l_{|S|},l_{t})$ such that $\sum_{1\leq i\leq |S|}l_i+l_{t}=n$, $l_i\geq 1$ $(1\leq i\leq |S|)$, and $l_{t}\geq 0$. Let $\mathbf{L}$ be the collection of all these integer sequences $\calL$.
\item For a sequence $\calL$, assume that $\Lambda_{\calL}=\tuple{G(U,V,E),\tuple{g_v}_{v\in U\cup V}}$ is a holant instance on a bipartite graph, where $U=\{u_1,\ldots,u_n\}$, and $V=S\cup \{t\}$ (we use vertex $t$ to represent the collection $T(S)$).
\item The weight function $w:E\rightarrow \R^{+}$ is defined as follows:
    \begin{enumerate}
    \item For a vertex $u_i\in U$ and a vertex $s_j\in S$, $w_{ij}=p_{ij}$.
    \item For a vertex $u_i\in U$ and $t\in V$, $w_{it}=\sum_{s_j\in T(S)}p_{ij}$.
    \end{enumerate}
\item For each vertex $u\in U$, the function $g_u=(=1)$.
    \footnote{Here the function $g_u=(=i)$ means that the function value $g_u$ is 1 if exactly $i$ edges incident to $u$ are of value 1 in the assignment. Otherwise, $g_u=0$}
    For each vertex $s_i\in S$, the function $g_{s_i}=(= l_i)$, and the function $g_{t}=(=l_{t})$.
\end{enumerate}

Since each $S\in \alg(\calP)$ is of constant size, we only need to enumerate at most $O(n^{|S|+1})=\poly(n)$ integer sequences $\calL$. Given an integer sequence $\calL=(l_1,\ldots,l_{|S|},l_{t})$, we can see that $Z(\Lambda_{\calL})$ is exactly the probability that $l_i$ nodes are realized at point $s_i\in S$ ($\forall 1\leq i\leq |S|$), and $l_t$ nodes are realized inside the point set $T(S)$. Then by Lemma~\ref{ob:prob}, we have the following equality:
$$
\Prob_{P\sim \calP}[\alg(P)=S]=\sum_{\calL\in \mathbf{L}}Z(\Lambda_{\calL}).
$$
It remains to show that we can compute each $Z(\Lambda_{\calL})$ efficiently. Fortunately, we have the following lemma.

\begin{lemma} (\cite{jerrum2004polynomial},\cite{tutte1954short})
\label{lm:degree}
For any bipartite graph $\Lambda_{\calL}$ with a specified integer sequence $\calL$, there exists an FPRAS to compute the holant value $Z(\Lambda_{\calL})$.
\end{lemma}

Thus, we have the following theorem.

\begin{theorem}
\label{thm:pras}
If both $k$ and $d$ are constants, given an instance of $n$ stochastic nodes in $\R^d$ in the locational uncertainty model, there exists a PTAS for the stochastic minimum $k$-center problem.
\end{theorem}

Combining Theorem~\ref{thm:existpras} and~\ref{thm:pras}, we obtain the main result Theorem~\ref{thm:praskcenter}.

\section{Stochastic Minimum $j$-Flat-Center}
\label{sec:coreset}
\eat{
\begin{definition}
\label{def:jflat}
Given a set $\calP$ of $n$ uncertainty points (in either locational or existential model) in $\R^d$. Given a $j$-flat $F\in \calF$ ($0\leq j\leq d-1$), where $\calF$ is the family of all $j$-flats in $\R^d$. We define the distance between $\calP$ and $F$ as follows:
$$
\flatdist(\calP,F)=\Exp_{P\sim \calP}[\max_{p\in P}\dist(p,F)].
$$
Minimum $j$-flat-center problem is to find a $j$-flat which minimizes the value $\flatdist(\calP,F)$.
\end{definition}

\begin{definition}
\label{def:jkproj}
Given a set $P$ of $n$ points in $\R^d$. Define a weight function $w:P\rightarrow \R^*$. For sme $0\leq j\leq d-1$ and $k\geq 1$, let $\calF$ be the family of shapes with each shape being a union of $k$ $j$-subspaces in $\R^d$. For a point $p\in \R^d$ and a shape $F\in \calF$, define $\dist(p,F)=\min_{q\in F}\dist(p,q)$. Let $\cost(P,F)=\sum_{s_i\in P}w_i\dist(s_i,F)$. A weighted $(j,k)$-projective clustering problem is to find a shape $F\in \calF$ which minimizes the value $\cost(P,F)$.
\end{definition}
}

In this section, we consider a generalized shape fitting problem, the minimum $j$-flat-center problem in the stochastic models. Let $\calF$ be the family of all $j$-flats in $\R^d$. Our main technique is to construct an \jflatcoreset\ of constant size, which satisfies that for any $j$-flat $F\in \calF$, we can use the \jflatcoreset\ to obtain a ($1\pm \e$)-estimation for the expected $j$-flat-center value $\flatdist(\calP,F)$. Then since the \jflatcoreset\ is of constant size, we have a polynomial system of constant size to compute the optimum in constant time.

Let $B=\sum_{1\leq i\leq n}p_i$ be the total probability. We discuss two different cases. If $B<\e$, we reduce the problem to a weighted $j$-flat-median problem, which has been studied in~\cite{varadarajan2012sensitivity}. If $B\geq\e$, the construction of an \jflatcoreset\ can be divided into two parts. We first construct a convex hull, such that with high probability (say $1-\e$) that all points are realized inside the convex hull. Then we construct a collection of point sets to estimate the contribution of points insider the convex hull. On the other hand, for the case that some point appears outside the convex hull, we again reduce the problem to a weighted $j$-flat-median problem. The definition of the weighted $j$-flat-median problem is as follows.

\begin{definition}
\label{def:jkproj}
For some $0\leq j\leq d-1$, let $\calF$ be the family of all $j$-flats in $\R^d$. Given a set $P$ of $n$ points in $\R^d$ together with a weight function $w:P\rightarrow \R^+$, denote $\cost(P,F)=\sum_{s_i\in P}w_i\cdot \dist(s_i,F)$. A weighted $j$-flat-median problem is to find a shape $F\in \calF$ which minimizes the value $\cost(P,F)$.
\end{definition}

\subsection{Case 1: $B<\e$}

In the first case, we show that the minimum $j$-flat-center problem can be reduced to a weighted $j$-flat-median problem. We need the following lemmas.

\begin{lemma}
\label{lm:sjflattoj1}
If $B<\e$, for any $j$-flat $F\in \calF$, we have $\sum_{s_i\in \calP}p_i\cdot\dist(s_i,F)\in (1\pm \e)\cdot\flatdist(\calP,F)$.
\end{lemma}

\begin{proof}
For a $j$-flat $F\in \R^d$, w.l.o.g., we assume that $\dist(s_i,F)$ is non-decreasing in $i$. Thus, we have
$$
\flatdist(\calP,F)=\sum_{i\in [n]}p_i\cdot \dist(s_i,F)\cdot\prod_{j>i}(1-p_j)
$$
Since $B<\e$, for any $i\in [n]$, we have that $1-\e\leq 1-\sum_{j\in [n]}p_i\leq \prod_{j>i}(1-p_j)\leq 1$. So we prove the lemma.
\end{proof}

By Lemma~\ref{lm:sjflattoj1}, we reduce the original problem to a weighted $j$-flat-median problem, where each point $s_i\in \calP$ has weight $p_i$.
We then need the following lemma to bound the total sensitivity.

\begin{lemma}(Theorem 18 in~\cite{varadarajan2012sensitivity})
\footnote{
	Theorem 18 in~\cite{varadarajan2012sensitivity} bounds the total sensitivity for the unweighted version. However, the proof can be extended to the weighted
version in a straightforward manner.}
\label{lm:totalsen}
Consider the weighted $j$-flat-median problem where $\calF$ is the set of all $j$-flats in $\R^d$. The total sensitivity of any weighted $n$-point set is $O(j^{1.5})$.
\end{lemma}

On the other hand, we know that the dimension of the weighted $j$-flat-median problem is $O(jd)$ by \cite{FL11}. Then by Lemma~\ref{lm:sentocore}, there exists an $\e$-coreset $\calS\subseteq \calP$ of cardinality $O(j^{4}d\e^{-2})$ to estimate the $j$-flat-median value $\sum_{s_i\in \calP}p_i\cdot\dist(s_i,F)$ for any $j$-flat $F\in \calF$.
\footnote{We remark that for the $j$-flat-median problem, Feldman and Langberg \cite{FL11} showed that there exists a coreset of size $O(jd\e^{-2})$. However, it is unclear how to generalize their technique to weighted version.}
Moreover, we can compute a constant approximation $j$-flat in $O(ndj^{O(j^2)})$ time by \cite{feldman2006coresets}. Then by \cite{varadarajan2012sensitivity}, we can construct an $\e$-coreseet $\calS$ in $O(ndj^{O(j^2)})$ time. Combining Lemma~\ref{lm:sjflattoj1}, we conclude the main lemma in this subsection.

\begin{lemma}
\label{lm:psmall}
Given an instance $\calP$ of $n$ stochastic points in $\R^d$, if the total probability $\sum_i p_i<\e$, there exists an \jflatcoreset\ of cardinality $O(j^{4}d\e^{-2})$ for the minimum $j$-flat-center problem. Moreover, we have an $O(ndj^{O(j^2)})$ time algorithm to compute the \jflatcoreset.
\end{lemma}

\eat{
\topic{Algorithm.}
\begin{enumerate}
\item Compute the optimal shape $F^*$. (It suffices to use a constant approximately optimal shape.) Compute $\calP'$ which is the projection of $\calP$ onto $F^*$.
\item Based on $\calP'$, compute a bound on the sensitivity of each point. By Lemma~\ref{lm:dimreduction}, we translate these sensitivities into a bound for $\sigma_{\calP}(s_i)$ for each $s_i\in \calP$.
\item Sample points from $\calP$ with probabilities proportional to $\sigma_{\calP}(s_i)$ to obtain a coreset.
\end{enumerate}
Therefore, the key point is how to find a constant approximately optimal shape $F^*$. By Theorem 4.4 in~\cite{deshpande2011algorithms}, the authors show that there is an algorithm to compute a constant approximately optimal shape $F^*$ with high probability via convex programming.
}

\subsection{Case 2: $B\geq \e$}

Note that if $F$ is a $j$-flat, the function $\dist(x,F)^2$ has a linearization. Here, a linearization is to map the function $\dist(x,F)^2$ to a $k$-variate linear function through variate embedding. The number $k$ is called the dimension of the linearization, see \cite{agarwal2005geometric}. We have the following lemma to bound the dimension of the linearization.

\begin{lemma} (\cite{feldman2013turning})
\label{lm:dimofl}
Suppose $F$ is a $j$-flat in $\R^d$, the function $\dist(x,F)^2$ ($x\in \R^d$) has a linerization. Let $D$ be the dimension of the linearization. If $j=0$, we have $D=d+1$. If $j=1$, we have $D=O(d^2)$. Otherwise, for $2\leq j\leq d-1$, we have $D=O(j^2d^3)$.
\end{lemma}

Suppose $\calP$ is an instance of $n$ stochastic points in $\R^d$. For each $j$-flat $F\in \R^d$, let $h_F(x)=\dist(x,F)^2$ ($x\in \R^d$), which admits a linearization of dimension $O(j^2d^3)$ by Lemma~\ref{lm:dimofl}. Now, we map each point $s\in \calP$ into an $O(j^2d^3)$ dimensional point $s'$ and map each $j$-flat $F\in \R^d$ into an $O(j^2d^3)$ dimensional direction $u$, such that $\dist(s,F)=\innerprod{s'}{u}^{1/2}$. For convenience, we still use $\calP$ to represent the collection of points after linearization. Recall that $\Prob[\vDash P]$ is the realized probability of the realization $P\sim \calP$. By this mapping, we translate our goal into finding a direction $u\in \R^{O(j^2d^3)}$, which minimizes the expected value $\Exp_{P\sim \calP}[\max_{x\in P}\innerprod{u}{x}^{1/2}]=\sum_{P\sim \calP}\Prob[\vDash P]\cdot \max_{x\in P}\innerprod{u}{x}^{1/2}$. We also denote $\calP^{\polar}=\{u\in \R^d\mid \innerprod{u}{s}\geq 0, \forall s\in \calP\}$ to be the polar set of $\calP$. We only care about the directions in the polar set $\calP^{\polar}$ for which $\innerprod{u}{s}^{1/2}$, $\forall s\in \calP$ is well defined.

We first construct a convex hull $\calH$ to partition the realizations into two parts. Our construction uses the method of \probkernel\ construction in \cite{huang2014epsilon}.  For any normal vector (direction) $u$, we move a sweep line $l_u$ orthogonal to $u$, along the direction $u$, to sweep through the points in $\calP$. Stop the movement of $\ell_u$ at the first point such that $\Prob[\calP\cap \bH_u)]\geq \e'$, where $\e'=\epsilon^{O(j^2d^3)}$ is a fixed constant. Denote $H_u$ to be the halfplane defined by the sweep line $\ell_u$ (orthogonal to the normal vector $u$)
and $\bH_u$ to be its complement. Denote $\calP(\bH_u)=\calP\cap \bH_u$ to be the set of points swept by the sweep line $l_u$. We repeat the above process for all normal vectors (directions) $u$, and let $\calH=\cap_u H_u$. Since the total probability $B\geq \e$, $\calH$ is nonempty by Helly's theorem. We also know that $\calH$ is a convex hull by \cite{huang2014epsilon}.
Moreover, we have the following lemma.

\begin{lemma} (Lemma 33 and Theorem 6 in \cite{huang2014epsilon})
\label{lm:lambdasumhigh}
Suppose the dimensionality is $d$. There is a convex set $\calK$, which is an intersection of $O(\epsilon^{-(d-1)/2})$ halfspaces and satisfies
$(1-\epsilon)\calK \subseteq \calH \subseteq \calK$. Moreover, $\calK$ can be constructed in $O(n \log^{O(d)}n)$ time.
\end{lemma}

By the above lemma, we construct a convex set $\calK=\cap_u \calK_{u}$, which is the intersection of $O(\epsilon^{-O(j^2d^3)})$ halfspaces $\calK_{u}$ ($u$ is the direction orthogonal to the halfspace $\calK_u$). Let $\bcalK_{u}$ be the complement of $\calK_u$, and let $\calP(\bcalK_u)=\calP\cap \bcalK_u$ be the set of points in $\bcalK_u$.
Denote $\calP(\bcalK)$ to be the set of points outside the convex set $\calK$.
Then we have the following lemma, which shows that the total probability outside $\calK$ is very small.

\begin{lemma}
\label{lm:outsmall}
Let $\calK$ be a convex set constructed as in Lemma~\ref{lm:lambdasumhigh}. The total probability $\Prob[\calP(\bcalK)]\leq \e$.
\end{lemma}

\begin{proof}
Assume that $\calK=\cap_u \calK_{u}$. Consider a halfspace $\calK_u$.
By Lemma~\ref{lm:lambdasumhigh}, the convex set $\calK$ satisfies that $\calH\subseteq \calK$. Thus, we have that $\Prob[\calP(\bcalK_u)]\leq \Prob[\calP(\bH_u)]\leq \e'$ by the definition of $\bH_u$.

Note that $\Prob[\calP(\bcalK)]$ is upper bounded by the multiplication of $\e'$ and the number of halfspaces of $\calK$. By Lemma~\ref{lm:lambdasumhigh}, there are at most $O(\epsilon^{-O(j^2d^3)})$ halfspaces $\calK_u$. Thus, we have that $\Prob[\calP(\bcalK)]\leq \e$.
\end{proof}

Our construction of \jflatcoreset\ is consist of two parts. For points inside $\calK$, we construct a collection $\calS_1$. Our construction is almost the same as \exprkernel\ construction in \cite{huang2014epsilon}, except that the cardinality of the collection $\calS_1$ is different.
For completeness, we provide the details of the construction here.
Let $\calP(\calK)$ be the collection of points in $\calK\cap \calP$, then $\calP(\calK)$ is also an instance of a stochastic minimum $j$-flat-center problem. We show that we can estimate $\Exp_{P\sim \calP(\calK)}[\max_{x\in P}\innerprod{u}{x}^{1/2}]$ by $\calS_1$. For the rest points outside $\calK$, we show that the contribution for the objective function $\Exp_{P\sim \calP}[\max_{x\in P}\innerprod{u}{x}^{1/2}]$ is almost linear and can be reduced to a weighted $j$-flat-median problem as in Case 1.

We first show how to construct $\calS_1$ for points inside $\calK$ as follows.

\begin{enumerate}
\item Sample $N=O((\e'\e)^{-2}\e^{-O(j^2d^3)}\log (1/\e))=O(\e^{-O(j^2d^3)})$ independent realizations restricted to $\calP(\calK)$.
\item For each realization $S_i$, use the algorithm in \cite{agarwal2004approximating} to find a
deterministic $\epsilon$-kernel $\calE_i$ of size $O(\epsilon^{-O(j^2d^3)})$. Here, a deterministic $\epsilon$-kernel $\calE_i$ satisfies that $(1-\e)CH(S_i)\subseteq CH(\calE_i)\subseteq CH(S_i)$, where $CH(\cdot)$ is the convex hull of the point set.
\item Let $\calS_1=\{\calE_i\mid 1\leq i \leq N\}$ be the collection of all $\e$-kernels, and each $\e$-kernel $\calE_i$ has a weight $1/N$.
\end{enumerate}

Hence, the total size of $\calS_1$ is $O(\e^{-O(j^2d^3)})$. For any direction $u\in \calP^{\polar}$, we use $\frac{1}{N}\sum_{\calE_i\in \calS_1} \max_{x\in \calE_i}\innerprod{u}{x}^{1/2}$
as an estimation of
$\Exp_{P\sim \calP(\calK)}[\max_{x\in P}\innerprod{u}{x}^{1/2}]$. By \cite{huang2014epsilon}, we have the following lemma.

\begin{lemma} (Lemma 38-40 in \cite{huang2014epsilon})
\label{lm:expr}
For any direction $u\in \calP^{\polar}$, let $M_u=\max_{x\in \calP(\calK)}\innerprod{u}{x}^{1/2}$. We have that
$$
\frac{1}{N}\sum_{\calE_i\in \calS_1} \max_{x\in \calE_i}\innerprod{u}{x}^{1/2}\in (1\pm \e/2)\Exp_{P\sim \calP(\calK)}[\max_{x\in P}\innerprod{u}{x}^{1/2}]\pm \epsilon'\e(1-\e) M_u/4
$$
\end{lemma}

Now we are ready to prove the following lemma.

\begin{lemma}
\label{lm:exprkernel}
For any direction $u\in \calP^{\polar}$, we have the following property.
$$
\frac{1}{N}\sum_{\calE_i\in \calS_1} \max_{x\in \calE_i}\innerprod{u}{x}^{1/2}+\sum_{s_i\in \calP(\bcalK)} p_i \cdot \innerprod{u}{s_i}^{1/2} \in (1\pm 4\e)\Exp_{P\sim \calP}[\max_{x\in P}\innerprod{u}{x}^{1/2}].
$$
\end{lemma}

\begin{proof}
Let $E$ be the event that no point is present in $\bcalK$. By the fact $\Prob[\bcalK]\leq \e$, we have that $\Prob[E]=\Pi_{s_i\in \calP(\bcalK)}(1-p_i)\geq 1-\sum_{s_i\in \calP(\bcalK)}p_i\geq 1-\e$. Thus, we conclude that $1-\e\leq \Prob[E]\leq 1$
We first rewrite $\Exp_{P\sim \calP}[\max_{x\in P}\innerprod{u}{x}^{1/2}]$ as follows:
\begin{align*}
&\Exp_{P\sim \calP}[\max_{x\in P}\innerprod{u}{x}^{1/2}]=\Prob[E]\cdot \Exp_{P\sim \calP}[\max_{x\in P}\innerprod{u}{x}^{1/2}\mid E]+\Prob[\overline{E}]\cdot\Exp_{P\sim \calP}[\max_{x\in P}\innerprod{u}{x}^{1/2}\mid \overline{E}] \\
=& \Prob[E]\cdot \Exp_{P\sim \calP(\calK)}[\max_{x\in P}\innerprod{u}{x}^{1/2}]+\Prob[\overline{E}]\cdot\Exp_{P\sim \calP}[\max_{x\in P}\innerprod{u}{x}^{1/2}\mid \overline{E}]
\end{align*}
For event $E$, we bound the term $\Prob[E]\cdot \Exp_{P\sim \calP(\calK)}[\max_{x\in P}\innerprod{u}{x}^{1/2}]$ via the collection $\calS_1$. Let $M_u=\max_{x\in \calP(\calK)}\innerprod{u}{x}^{1/2}$. By Lemma~\ref{lm:expr}, for any direction $u\in \calP^{\polar}$, we have that
$$
\frac{1}{N}\sum_{\calE_i\in \calS_1} \max_{x\in \calE_i}\innerprod{u}{x}^{1/2}\in (1\pm \e/2)\Exp_{P\sim \calP(\calK)}[\max_{x\in P}\innerprod{u}{x}^{1/2}]\pm \epsilon'\e(1-\e) M_u/4
$$
By Lemma~\ref{lm:lambdasumhigh}, we have that $(1-\e)\calK\subseteq \calH$. Then by the construction of $\calH_u$, we have that $\Prob[\calP\cap (1-\e)\bcalK_u]\geq \e'$. Thus, we obtain that
$$
\Exp_{P\sim \calP}[\max_{x\in P}\innerprod{u}{x}^{1/2}]\geq \e'(1-\e)\max_{x\in \calP(\calK)}\innerprod{u}{x}^{1/2}= \e'(1-\e) M_u.
$$
So we conclude that
\begin{equation}
\label{eq:1}
\begin{split}
&(1-2\e)\Prob[E]\cdot \Exp_{P\sim \calP(\calK)}[\max_{x\in P}\innerprod{u}{x}^{1/2}]- \e\Exp_{P\sim \calP}[\max_{x\in P}\innerprod{u}{x}^{1/2}]\leq \frac{1}{N}\sum_{\calE_i\in \calS_1} \max_{x\in \calE_i}\innerprod{u}{x}^{1/2} \\
\leq &(1+ 2\e)\Prob[E]\cdot \Exp_{P\sim \calP(\calK)}[\max_{x\in P}\innerprod{u}{x}^{1/2}]+ \e\Exp_{P\sim \calP}[\max_{x\in P}\innerprod{u}{x}^{1/2}],
\end{split}
\end{equation}
since $1-\e\leq \Prob[E]\leq 1$.

For event $\overline{E}$, without loss of generality, we assume that the $n$ points $s_1,\ldots,s_n$ in $\calP$ are sorted in nondecreasing order according to the inner product $\innerprod{u}{s_i}$. Assume that $s_{i_1},\ldots,s_{i_l}$ ($i_1<i_2<\ldots<i_l$) are points in $\calP(\bcalK)$. Let $E_j$ be the event that point $s_{i_j}$ is present and all points $s_{i_k}$ are not present for $k>j$.
\eat{
Denote $E^+_j$ to be the event that point $s_{i_j}$ is the farthest point according to direction $u$, i.e., $\innerprod{u}{s_{i_j}}$ is the maximum. Note that event $E^+_j$ implies event $E_j$. However, event $E_j$ does not imply event $E^+_j$, since the farthest point according to direction $u$ might be in $\calP(\calK)$. Let $E^-_j$ be the event that $E_j$ happens and $E^+_j$ does not happen.}
We have that
\begin{align*}
& \Prob[\overline{E}]\cdot\Exp_{P\sim \calP}[\max_{x\in P}\innerprod{u}{x}^{1/2}\mid \overline{E}] = \sum_{j\in [l]}\Prob[E_j]\cdot\Exp_{P\sim \calP}[\max_{x\in P}\innerprod{u}{x}^{1/2}\mid E_j]\\
=&\sum_{j\in [l]}p_{i_j}\cdot\bigl(\prod_{j+1\leq k\leq l}(1-p_{i_k})\bigr)\cdot\Exp_{P\sim \calP}[\max_{x\in P}\innerprod{u}{x}^{1/2}\mid E_j].
\end{align*}
By the above equality, on one hand, we have that
\begin{equation}
\label{eq:2}
\Prob[\overline{E}]\cdot\Exp_{P\sim \calP}[\max_{x\in P}\innerprod{u}{x}^{1/2}\mid \overline{E}]\geq (1-\e)\sum_{j\in [l]}p_{i_j}\cdot \innerprod{u}{s_{i_j}}^{1/2},
\end{equation}
since $\max_{x\in P}\innerprod{u}{x}^{1/2}\geq \innerprod{u}{s_{i_j}}^{1/2}$ if event $E_j$ happens.
On the other hand, the following inequality also holds.
\begin{equation}
\label{eq:3}
\begin{split}
&\Prob[\overline{E}]\cdot\Exp_{P\sim \calP}[\max_{x\in P}\innerprod{u}{x}^{1/2}\mid \overline{E}]= \sum_{j\in [l]}\Prob[E_j]\cdot\Exp_{P\sim \calP}[\max_{x\in P}\innerprod{u}{x}^{1/2}\mid E_j] \\
\leq &\sum_{j\in [l]}\Prob[E_j]\cdot\Exp_{P\sim \calP}[\innerprod{u}{s_{i_j}}^{1/2}+\max_{x\in P\cap \calP(\calK)}\innerprod{u}{x}^{1/2}\mid E_j] \\
\leq &\sum_{j\in [l]}p_{i_j}\cdot \bigl(\Exp_{P\sim \calP}[\innerprod{u}{s_{i_j}}^{1/2}\mid E_j]+\Exp_{P\sim \calP}[\max_{x\in P\cap \calP(\calK)}\innerprod{u}{x}^{1/2}\mid E_j]\bigr) \\
\leq &\sum_{j\in [l]}p_{i_j}\cdot \innerprod{u}{s_{i_j}}^{1/2}+ \sum_{j\in [l]}p_{i_j}\cdot \Exp_{P\sim \calP(\calK)}[\max_{x\in P}\innerprod{u}{x}^{1/2}]
\leq \sum_{j\in [l]}p_{i_j}\cdot \innerprod{u}{s_{i_j}}^{1/2}+ \e\cdot \Exp_{P\sim \calP}[\max_{x\in P}\innerprod{u}{x}^{1/2}].
\end{split}
\end{equation}
The last inequality holds since that $\sum_{j\in [l]}p_{i_j}=\Prob[\calP(\bcalK)]\leq \e$ by Lemma~\ref{lm:outsmall}. Combining Inequalities \eqref{eq:1}, \eqref{eq:2} and \eqref{eq:3}, we prove the lemma.
\end{proof}

\eat{
Note that our construction of coreset $\calS_1$ is under the subset constraint. Thus we can map $\calS_1$ to original points and obtain the following lemma.

\begin{lemma}
\label{lm:originexprkernel}
For any $j$-flat in $\R^d$, we have the following property.
$$
\Exp_{P\sim \calS_1}[\max_{x\in P}\dist(x,F)]+\sum_{s_i\in \calP\cap \overline{\calH}} p_i \cdot \dist(s_i,F) \in (1\pm 4\e)\Exp_{P\sim \calP}[\max_{x\in P}\dist(x,F)].
$$
\end{lemma}
}

By Lemma~\ref{lm:psmall}, we construct a point set $\calS_2$ to estimate $\sum_{s_i\in \calP(\bcalK)} p_i \cdot \dist(s_i,F)$ with a weight function $w':\calS_2\rightarrow \R$. We have that the size of $\calS_2$ can be bounded by $O(j^4d\e^{-2})$. Then $\calS=\calS_1\cup \calS_2$ is a collection of constant size, which satisfies the following property:
\begin{equation}
\label{ieq:core}
\frac{1}{N}\sum_{\calE_i\in \calS_1} \max_{x\in \calE_i}\innerprod{u}{x}^{1/2}+\sum_{s_i\in \calS_2} w'_i \cdot \innerprod{u}{s_i}^{1/2} \in (1+ O(\e))\Exp_{P\sim \calP}[\max_{x\in P}\innerprod{u}{x}^{1/2}].
\end{equation}
Here $w'_i$ is the weight of $s_i$ in $\calS_2$. We can think $\calS_2=\{\{s_i\}\mid 1\leq s_i\leq  |\calS_2|\}$ as a collection of singleton point sets $\{s_i\}$. Then by Inequality~\ref{ieq:core}, we have that $\calS$ is a generalized $\e$-coreset satisfying Definition~\ref{def:core}. We conclude the following lemma.
\begin{lemma}
\label{lm:plarge}
Given an instance $\calP$ of $n$ stochastic points of the stochastic minimum $j$-flat-center problem in the existential model, if the total probability $\sum_i p_i\geq\e$, there exists an \jflatcoreset\ $\calS$ containing $O(\e^{-O(j^2d^3)}+j^4d\e^{-2})$ point sets of size at most $O(\e^{-O(j^2d^3)})$, together with a weight function $w': \calS\rightarrow \R^+$, which satisfies that for any $j$-flat $F\in \calF$,
$$
\sum_{S\in \calS}w'(S)\cdot\flatdist(S,F) \in (1\pm \e)\flatdist(\calP,F).
$$
\end{lemma}

Combining Lemma~\ref{lm:psmall} and Lemma~\ref{lm:plarge}, we can obtain the following theorem.

\begin{theorem}
\label{thm:j1exist}
Given an instance $\calP$ of $n$ stochastic points in the existential model, there is an \jflatcoreset\ of size $O(\e^{-O(j^2d^3)}+j^4d\e^{-2})$ for the minimum $j$-flat-center problem. Moreover, we have an $O(n \log^{O(d)}n+\e^{-O(j^2d^3)} n)$ time algorithm to compute the \jflatcoreset.
\end{theorem}

\begin{proof}
We only need to prove the running time. Recall that the \jflatcoreset\ $\calS$ can be divided into two parts $\calS=\calS_1\cup \calS_2$. For the first part $\calS_1$, we construct the convex hull $\calK$ in $O(n \log^{O(d)}n)$ by Lemma~\ref{lm:lambdasumhigh}. Then we construct $\calS_1$ by taking $O(\e^{-O(j^2d^3)})$ independent realizations restricted to $\calP(\calK)$. For each sample, we construct a deterministic $\e$-kernel in $O(n+\e^{-(d-3/2)})$ by \cite{chan2004faster, yu2008practical}. So the total time for constructing $\calS_1$ is $O(n \log^{O(d)}n+\e^{-O(j^2d^3)} n)$. On the other hand, we can construct $\calS_2$ in $O(ndj^{O(j^2)})$ time by Lemma~\ref{lm:psmall}. Thus, we prove the theorem.
\end{proof}

\topic{PTAS for stochastic minimum $j$-flat-center.} Given an \jflatcoreset\ $\calS$ together with a weight function $w': \calS\rightarrow \R^+$ by Theorem~\ref{thm:j1exist}, it remains to show how to compute the optimal $j$-flat for $\calS$. Our goal is to find the optimal $j$-flat $F^*$ such that the total generalized distance $\sum_{S\in \calS}w'(S)\cdot\flatdist(S,F^*)$ is minimized. The argument is similar to the stochastic minimum $k$-center problem.

We first divide the family $\calF$ of $j$-flats into a constant number of sub-families. In each sub-family $\calF'\subseteq \calF$, we have the following property: for each $S_i\in \calS$, and each $j$-flat $F\in \calF'$, the point $s^i=\arg \max_{s\in S_i} \dist(s,F)$ is fixed.
By Lemma 41, we have that $h_F(x)=\dist(x,F)^2$ ($x\in \R^d$) admits a linearization of dimension $O(j^2 d^3)$. For each sub-family $\calF'$, we can formulate the optimization problem as a polynomial system of constant degree, a constant number of variables, and a constant number of constraints. Then we can compute the optimal $j$-flat in constant time for all sub-families $\calF'\subseteq \calF$. Thus, we can compute the optimal $j$-flat-center for the \jflatcoreset\ $\calS$ in constant time. We then have the following corollary.

\begin{corollary}
\label{cor:j1existpras}
If the dimensionality $d$ is a constant, given an instance of $n$ stochastic points in $\R^d$ in the existential uncertainty model, there exists a PTAS for the stochastic minimum $j$-flat-center problem in $O(n \log^{O(d)}n+\e^{-O(j^2d^3)} n)$ time.
\end{corollary}

\topic{Locational Uncertainty Model} Note that in the locational uncertainty model, we only need to consider Case 2. We use the same construction as in the existential model. Let $p_i=\sum_{j}p_{ji}$. Similarly, we make a linearization for the function $\dist(x,F)^2$, where $x\in\R^d$ and $F\in \calF$ is a $j$-flat. Using this linearization, we also map $\calP$ into $O(j^2d^3)$-dimensional points. For the $j$th node and a set $P$ of points, we denote $p_j(P)=\sum_{s_i\in P}p_{ji}$ to be the total probability that the $j$th node locates inside $P$.

By the condition $\Prob[\calP(\bcalK)]\leq \e$, we have that $\Prob[E]=1- \prod_{j\in [m]}(1-p_j(\bcalK))\leq 1-(1-\e)=\e $, where event $E$ represents that there exists a point present in $\bcalK$. So we can regard those points outside $\calK$ independent. On the other hand, for any direction $u$, since $\Prob[\calP\cap (1-\e)\overline{\calH}_u]\geq \e'$, we have that $\Prob[E_u]=1-\prod_{j\in [m]}(1-p_j(\calP\cap (1-\e)\bH_u))\geq 1-(1-\frac{\e'}{m})^m\geq \e'/2$, where event $E_u$ represents that there exists a point present in $\calP\cap (1-\e)\bH_u$. Moreover, we can use the same method to construct a collection $\calS_1$ as an estimation for the point set $\calP(\calK)$ in the locational uncertainty model. So Lemma~\ref{lm:exprkernel} still holds. Then by Lemma~\ref{lm:plarge}, we can construct an \jflatcoreset\ of constant size.

\begin{theorem}
\label{thm:j1location}
Given an instance $\calP$ of $n$ stochastic points in the locational uncertainty model, there is an \jflatcoreset\ of cardinality $O(\e^{-O(j^2d^3)}+j^4d\e^{-2})$ for the minimum $j$-flat-center problem. Moreover, we have a polynomial time algorithm to compute the gerneralized $\e$-coreset.
\end{theorem}

By a similar argument as in the existential model, we can give a PTAS for the locational uncertainty model. Then combining with Corollary~\ref{cor:j1existpras}, we prove the main result Theorem~\ref{thm:prasjflat}.

\eat{
\topic{Remark:} We can generalize the above result by considering other distance functions. In fact, we can solve a class of distance functions $\dist(\cdot,\cdot)=\|\cdot\|^z$ for $z\geq 1$. Note that the Euclidean distance function is a special case when $z=1$.
\jian{remove the remark. it generates confusions.}
}

\bibliographystyle{plain}
\bibliography{shapefitting}

\appendix
\section{Proof of Lemma~\ref{lm:sentocore}}
\label{app:pre}

The following theorem is a restatement of Theorem 4.1 
and its proof in \cite{FL11}. 
Lemma~\ref{lm:sentocore} is a direct corollary from the following theorem.

\begin{theorem}
\label{thm:fl}
Let $D=\{g_i\mid 1\leq i\leq n\}$ be a set of $n$ functions. For each $g\in D$, $g: X\rightarrow \R^{\geq 0}$ is a function from a ground set $X$ to $[0,+\infty)$. Let $0<\e<1/4$ be a constant. Let $m:D\rightarrow \R^+$ be a function on $D$ such that
\begin{equation}
\label{ieq:sen}
q(g)\geq \max_{x\in X}\frac{g(x)}{\sum_{g\in D}g(x)}.
\end{equation}
Then there exists a collection $\calS\subseteq D$ of functions, together with a weight function $w': \calS\rightarrow \R^+$, such that for every $x\in X$
$$
|\sum_{g\in D}g(x)-\sum_{g\in \calS}w'(g)\cdot g(x)|\leq \e \sum_{g\in Y}g(x),
$$
Moreover, the size of $\calS$ is 
$$
O\left(\left(\frac{\sum_{g\in D}q(g)}{\e}\right)^2 \dim(D)\right),
$$
where $\dim(D)$ is the generalized shattering dimension of $D$ (see Definition 7.2 in \cite{FL11}).
\end{theorem}D
Now we are ready to prove Lemma~\ref{lm:sentocore}.

\begin{proof}
Suppose that we are given a (weighted) instance $\boldS=\{S_i\mid S_i\subset \R^d, 1\leq i\leq n\}$ of a generalized shape fitting problem $(\R^d,\calF,\fun)$, with a weight function $w:\boldS\rightarrow \R^+$. A generalized $\e$-coreset is a collection $\calS\subseteq \boldS$ of point sets, together with a weight function $w':\calS\rightarrow \R^+$ such that, for any shape $F\in \calF$, we have
\begin{equation}
\label{ieq:app}
\sum_{S_i\in \calS}w'_i \cdot \fun(S_i,F)\in
(1\pm \e)\sum_{S_i\in \boldS}w_i\cdot \fun(S_i,F).
\end{equation}
For every $S_i\in \boldS$ and $F\in \calF$, let $g_{i}(F)=w_i\cdot \fun(S_i,F)$ and $D=\{g_{i}\mid S_i\in \boldS\}$. Define
$$
q(g_{i})=\sigma_{\boldS}(S_i)+\frac{1}{n}=\inf\{\beta\geq 0\mid w_i\cdot \fun(S_i,F)\leq \beta \cdot \sum_{S_i\in \boldS}w_i\cdot \fun(S_i,F), \forall F\in \calF\}+\frac{1}{n}.
$$
It is not hard to verify that this definition satisfies Inequality~\ref{ieq:sen}.
The additional $1/n$ term will be useful in Appendix~\ref{app:core},
where we need a lower bound of $q(g_{i})$.
Thus, we have $\mathfrak{G}_{\boldS}+1=\sum_{S_i\in \boldS}(\sigma_{\boldS}(S_i)+1/n)=\sum_{g_i\in D}q(g_i)$. Recall that $\dim(\boldS)$ is the generalized shattering dimension of $\boldS$. By Theorem~\ref{thm:fl}, we conclude that there exists a collection $\calS$ of cardinality $O\bigl((\frac{\mathfrak{G}_{\boldS}}{\e})^2 \dim(\boldS)\bigr)$ with a weight function $w':\calS\rightarrow \R^+$ satisfying Inequality~\eqref{ieq:app}.
\end{proof}

\section{Constructing additive $\e$-coresets}
\label{app:kcenter}

In this section, we first give the algorithm for constructing an additive $\e$-coreset. We construct Cartesian grids and maintain one point from each nonempty grid cell, which is similar to \cite{agarwal2002exact}. However, our algorithm is more complicated. See Algorithm~\ref{alg:core} for details.

\begin{algorithm}
\caption{Constructing additive $\e$-coresets (\ACORESET)}
\label{alg:core}
Input: a realization $P\sim \calP$. W.l.o.g., assume that $P=\{s_1,\ldots,s_m\}$. \\
Let $r_P= \min_{F:F\subseteq \calP, |F|=k}\maxdist(P,F)$. If $r_P=0$, output $\alg(P)=P$. Otherwise assume that $2^a\leq r_P<2^{a+1}$ ($a\in \calZ$). \\
Draw a $d$-dimensional Cartesian grid $G_1(P)$ of side length $\e 2^{a}/4d$ centered at point $0^d$. \\
Let $\calC_1(P)=\{C\mid C\in G, C\cap P\neq \emptyset\}$ be the collection of those cells which intersects $P$. \\
For each cell $C\in \calC_1(P)$, let $s^C\in C\cap P$ be the point in $C$ of smallest index. Let $\alg_1(P)= \{s^C\mid C\in \calC_1(P)\}$. \\
Compute $r_{\alg_1(P)}= \min_{F:F\subseteq \calP, |F|=k}\maxdist(\alg_1(P),F)$. If $r_{\alg_1(P)}\geq 2^{a}$, let $\alg(P)=\alg_1(P)$, $G(P)=G_1(P)$, and $\calC(P)=\calC_1(P)$.\\
If $r_{\alg_1(P)}< 2^a$, draw a $d$-dimensional Cartesian grid $G_2(P)$ of side length $\e 2^{a}/8d$ centered at point $0^d$. Repeat step 4 and 5, construct $\calC_2(P)$ and $\alg_2(P)$ based on the new Cartesian grid $G_2(P)$. Let $\alg(P)=\alg_2(P)$, $G(P)=G_2(P)$, and $\calC(P)=\calC_2(P)$.\\
Output $\alg(P)$, $G(P)$, and $\calC(P)$.
\end{algorithm}

Now we analyze the algorithm.

\begin{observation}
\label{ob:1}
$r_P$ is a 2-approximation for the minimum $k$-center problem w.r.t. $P$.
\end{observation}

\begin{proof}
By Gonzalez's greedy algorithm~\cite{gonzalez1985clustering}, there exists a subset $F\subseteq P\subseteq \calP$ of size $k$ such that the $k$-center value $\maxdist(P,F)$ is a 2-approximation for the minimum $k$-center problem w.r.t. $P$. Thus, we prove the observation.
\end{proof}

By the above observation, we have the following lemma.

\topic{Lemma
\ref{ob:2}.}
The running time of \ACORESET\ on any $n$ point set $P$ is $O(kn^{k+1})$. Moreover, the output $\alg(P)$ is an additive $\e$-coreset of $P$ of size at most $O(k/\e^d)$.

\begin{proof}
Since $r_P$ is a 2-approximation, $\alg(P)$ is an additive $\e$-coreset of $P$ of size $O(k/\e^d)$ by Theorem 2.4 in~\cite{agarwal2002exact}. For the running time, consider computing $r_P$ in Step 2 (also $r_{\alg_1(P)}$ in Step 6). There are at most $n^k$ point sets $F\subseteq \calP$ such that $|F|=k$. Note that computing $\maxdist(P,F)$ costs at most $nk$ time. Thus, it costs $O(kn^{k+1})$ time to compute $r_P$ (also $r_{\alg_1(P)}$) for all $k$-point sets $F\subseteq \calP$. On the other hand, it only costs linear time to construct the Cartesian grid $G(P)$, the cell collection $\calC(P)$ and $\alg(P)$ after computing $r_P$ and $r_{\alg_1(P)}$, which finishes the proof.
\end{proof}

We then give the following lemmas, which is useful for proving Lemma~\ref{ob:prob}.

\begin{lemma}
\label{ob:monotone}
For two point sets $P,P'$, if $P'\subseteq P$, then $r_{P'}\leq r_{P}$. Moreover, if $P'$ is an additive $\e$-coreset of $P$, then $(1-\e)r_P\leq r_{P'}\leq r_P$.
\end{lemma}

\begin{proof}
Suppose $F\subseteq \calP$ is the $k$-point set such that the $k$-center value $\maxdist(P,F)=r_P$. Since $P'\subseteq P$, we have $\maxdist(P',F)\leq r_P$. Thus, we have $r_{P'}\leq \maxdist(P',F)\leq r_{P}$.

Moreover, assume that $P'$ is an additive $\e$-coreset of $P$. Suppose $F'\subseteq \calP$ is the $k$-point set such that the $k$-center value $\maxdist(P',F')=r_{P'}$. Then by Definition~\ref{def:addcore}, we have $\maxdist(P,F')\leq (1+\e)r_{P'}$. Thus, we have $(1-\e)r_P\leq (1-\e)\maxdist(P,F')<r_{P'}\leq r_P$.
\end{proof}

\begin{lemma}
\label{ob:same}
Assume that a point set $P'=\alg(P)$ for another point set $P\sim \calP'$. Running \ACORESET$(P')$ and \ACORESET$(P)$, assume that we obtain two Cartesian grids $G(P')$ and $G(P)$ respectively. Then we have $G(P')=G(P)$.
\end{lemma}

\begin{proof}
If $r_P=0$, we have that $r_{P'}\leq r_P=0$ by Lemma~\ref{ob:monotone}. Thus we do not construct the Cartesian grid for both $P$ and $P'$.
Otherwise assume that $2^a\leq r_P<2^{a+1}$ ($a\in \calZ$). Run \ACORESET$(P)$. In Step 5, we construct a Cartesian grid $G_1(P)$ of side length $\e2^{a}/4d$, a cell collection $\calC_1(P)$, and a point set $\alg_1(P)$. Since $\alg_1(P)$ is an additive $\e$-coreset of $P$ by~\cite{agarwal2002exact}, we have $2^{a-1}<(1-\e)r_P\leq r_{\alg_1(P)}\leq r_S<2^{a+1}$. Then we consider the following two cases.

Case 1: $r_{\alg_1(P)}\geq 2^a$.
Then $P'=\alg(P)=\alg_1(P)$, and $G(P)=G_1(P)$ in this case. Running \ACORESET$(P')$, we have that $2^a\leq r_{\alg_1(P)}=r_{P'}\leq r_{P}<2^{a+1}$ by Lemma~\ref{ob:monotone}. Thus, we construct a Cartesian grid $G_1(P')=G_1(P)$ of side length $\e2^{a}/4d$, and a point set $\alg_1(P')$ in Step 5. Since $G_1(P')=G_1(P)$ and $P'=\alg_1(P)$, we have that $\alg_1(P')=P'$ by the construction of $\alg_1(P')$. Thus, $r_{\alg_1(P')}=r_{\alg_1(P)}\geq 2^a$, and we obtain that $G(P')=G_1(P')$ in Step 6, which proves the lemma.

Case 2: $2^{a-1}\leq r_{\alg_1(P)}< 2^a$.
Then in Step 7, we construct a Cartesian grid $G_2(P)$ of side length $\e2^{a}/8d$ for $P$, a cell collection $\calC_2(P)$, and a point set $\alg_2(P)$. In this case, we have that $\alg(P)=\alg_2(P)$, $G(P)=G_2(P)$, and $\calC(P)=\calC_2(P)$.
Now run \ACORESET$(P')$, and obtain $\alg(P')$, $G(P')$, and $\calC(P')$. By Lemma~\ref{ob:monotone}, we have
$$
2^{a+1}>r_{P}\geq r_{P'}=r_{\alg(P)}\geq (1-\e)r_{P}> 2^{a-1}.
$$
We need to consider two cases. If $2^{a-1}\leq r_{P'}<2^a$, we construct a Cartesian grid $G_1(P')$ of side length $\e2^{a}/8d$, and a point set $\alg_1(P')$ in Step 5. Since $G_1(P')=G_2(P)$ and $P'=\alg_2(P)$, we have that $\alg_1(P')=P'$ by the construction of $\alg_1(P')$. Then we let $G(P')=G_1(P')$ in Step 6. In this case, both $G(P)$ and $G(P')$ are of side length $\e2^{a}/8d$, which proves the lemma.

Otherwise if $2^a\leq r_{P'}< 2^{a+1}$, we construct the Cartesian grid $G_1(P')=G_1(P)$ of side length $\e2^{a}/4d$, a cell collection $\calC_1(P')$, and a point set $\alg_1(P')$ in Step 5. We then prove that $\alg_1(P')=\alg_1(P)$. Since all Cartesian grids are centered at point $0^d$, a cell in $G_1(P)$ can be partitioned into $2^d$ equal cells in $G_2(P)$. Rewrite a cell $C^*\in G_1(P)$ as $C^*=\cup_{1\leq i\leq 2^d}C_i$ where each $C_i\in G_2(P)$. Assume that point $s^*\in C^*\cap P=\cup_{1\leq i\leq 2^d}(C_i\cap P)$ has the smallest index, then point $s^*$ is also the point in $C^*\cap \alg_2(P)$ of smallest index. Since $\alg(P)=\alg_2(P)$, we have that
$s^*$ is the point in $C^*\cap \alg(P)$ of smallest index.
Considering $\alg_1(P')$, note that for each cell $C^*\in \calC_1(P')$, $\alg_1(P')$ only contains the point in $C^*\cap P'$ of smallest index. Since $P'=\alg(P)$, we have that $\alg_1(P')=\alg_1(P)$. Thus, we conclude that $r_{\alg_1(P')}=r_{\alg_1(P)}<2^a$. Then in Step 7, we construct a Cartesian grid $G_2(P')=G_2(P)$ of side length $\e2^{a}/8d$ for $P'$. Finally, we output $G(P')=G_2(P')=G(P)$, which proves the lemma.
\end{proof}

Recall that we denote $\alg(\calP)=\{\alg(P)\mid P\sim \calP\}$ to be the collection of all possible additive $\e$-coresets. For any $S$, we denote $\alg^{-1}(S)=\{P\sim \calP \mid \alg(P)=S\}$ to be the collection of all realizations mapped to $S$. Now we are ready to prove Lemma~\ref{ob:prob}.

\topic{Lemma~\ref{ob:prob}.} (restated)
Consider a subset $S$ of at most $O(k/\e^d)$ points.
Run algorithm \ACORESET$(S)$, which outputs
an additive $\e$-coreset $\alg(S)$, a Cartesian grid $G(S)$, and
a collection $\calC(S)$ of nonempty cells.
If $\alg(S)\neq S$, then $S\notin \alg(\calP)$
(i.e., $S$ is not the output of \ACORESET\ for any realization $P\sim \calP$).
If $|S|\leq k$, then $\alg^{-1}(S)=\{S\}$.
Otherwise if $\alg(S)=S$ and $|S|\geq k+1$,
then a point set $P\sim \calP$ satisfies $\alg(P)=S$ if and only if
\begin{enumerate}
\item[P1.] For any cell $C\notin \calC(S)$, $C\cap P=\emptyset$.
\item[P2.] For any cell $C\in \calC(S)$, assume that point $s^C=C\cap S$.
Then $s^C\in P$, and any point $s'\in C\cap \calP$ with a smaller index
than that of $s^C$ does not appear in the realization $P$.
\end{enumerate}

\begin{proof}
If $\alg(S)\neq S$, we have that $r_S>0$. Assume that $S\in \alg(\calP)$. There must exist some point set $P\sim \calP$ such that $\alg(P)=S$. By Lemma~\ref{ob:same}, running \ACORESET$(P)$ and \ACORESET$(S)$, we obtain the same Cartesian grid $G(P)=G(S)$. Since $\alg(S)\neq S$, there must exist a cell $C\in \calC(S)$ such that $|C\cap S|\geq 2$ (by the construction of $\alg(S)$). Note that $C\in G(P)$. We have $|C\cap \alg(P)|=1$, which is a contradiction with $\alg(P)=S$. Thus, we conclude that $S\notin \alg(\calP)$.

If $|S|\leq k$, assume that there exists another point set $P\neq S$, such that $\alg(P)=S$. By Lemma~\ref{ob:2}, we know that $S$ is an additive $\e$-coreset of $P$. By Definition~\ref{def:addcore}, we have $S\subseteq P$ and $\maxdist(P,S)\leq (1+\e)\maxdist(S,S)=0$. Thus we conclude that $P=S$. On the other hand, we have $\alg(S)=S$ since $r_S=0$. So we conclude that $\alg^{-1}(S)=\{S\}$.

If $|S|\geq k+1$ and $\alg(S)=S$, we have that $r_S>0$. Running \ACORESET$(P)$ and \ACORESET$(S)$, assume that we obtain two Cartesian grids $G(P)$ and $G(S)$ respectively. By Lemma~\ref{ob:same}, if $\alg(P)=S$, then we have $G(P)=G(S)$. Moreover, by the construction of $\alg(P)$, P1 and P2 must be satisfied.

We then prove the 'only if' direction. If P1 and P2 are satisfied, we have that $S$ is an additive $\e$-coreset of $P$ satisfying Definition~\ref{def:addcore} by~\cite{agarwal2002exact}. Then by Lemma~\ref{ob:monotone}, we have that $(1-\e)r_P\leq r_S\leq r_P$.
Assume that $2^a\leq r_S<2^{a+1}$ ($a\in \calZ$), we conclude $2^{a}\leq r_P<2^{a+2}$. Now run \ACORESET$(S)$. In Step 5, assume that we construct a Cartesian grid $G_1(S)$ of side length $\e2^{a}/4d$, a cell collection $\calC_1(S)$, and a point set $\alg_1(S)$. Since $\alg_1(S)$ is an additive $\e$-coreset of $S$ by~\cite{agarwal2002exact}, we have $2^{a-1}<(1-\e)r_S\leq r_{\alg_1(S)}\leq r_S<2^{a+1}$. Then we consider the following two cases.

Case 1: $2^{a}\leq r_{\alg_1(S)}< 2^{a+1}$. In this case, we have that $G(S)=G_1(S)$, $\calC(S)=\calC_1(S)$, and $S=\alg(S)=\alg_1(S)$. Running \ACORESET$(P)$, assume that we obtain $G(P)$, $\calC(P)$, and $\alg(P)$. Consider the following two cases. If $2^a\leq r_P< 2^{a+1}$, we construct a Cartesian grid $G_1(P)=G(S)$ of side length $\e2^{a}/4d$, and a point set $\alg_1(P)$ in Step 5. Since P1 and P2 are satisfied, we know that $\alg_1(P)=S$. Then since $2^a\leq r_{\alg_1(P)}=r_S<2^{a+1}$, we obtain that $\alg(P)=\alg_1(P)=S$ in this case. Otherwise if $2^{a+1}\leq r_S< 2^{a+2}$, run \ACORESET$(P)$. We construct a Cartesian grid $G_1(P)$ of side length $\e2^{a}/2d$, and a point set $\alg_1(P)$ in Step 5. Since P1 and P2 are satisfied, we have that $\alg_1(P)\subseteq S$. Thus, we have $r_{\alg_1(P)}\leq r_S<2^{a+1}$ by Lemma~\ref{ob:monotone}. Then in Step 7, we construct a Cartesian grid $G_2(P)=G_1(S)$ of side length $\e2^{a}/4d$, and a point set $\alg_2(P)$. In this case, we have that $G(P)=G_2(P)=G_1(S)$, and $\alg(P)=\alg_2(P)$. By P1 and P2, we have that $\alg(P)=\alg_2(P)=S$.

Case 2: $2^{a-1}\leq r_{\alg_1(S)}< 2^a$. Running \ACORESET$(S)$, we construct a Cartesian grid $G_2(S)$ of side length $\e2^{a}/8d$, and a point set $\alg_2(S)$ in Step 7. In this case, we have that $G(S)=G_2(S)$, and $S=\alg(S)=\alg_2(S)$. Since $\alg_1(S)$ is an additive $\e$-coreset of $S$, we conclude that $\alg_1(S)$ is also an additive $3\e$-coreset of $P$ satisfying Definition~\ref{def:addcore}. Then we have that $2^{a}\leq r_P\leq (1+3\e)r_{\alg_1(S)}< 2^{a+1}$ by Lemma~\ref{ob:monotone}. Running \ACORESET$(P)$, we construct a Cartesian grid $G_1(P)=G_1(S)$ of side length $\e2^{a}/4d$, and a point set $\alg_1(P)$ in Step 5. Since P1 and P2 are satisfied, we know that $\alg_1(P)= \alg_1(S)$. Thus, we have $2^{a-1}\leq r_{\alg_1(P)}=r_{\alg_1(S)}<2^a$. Then in Step 7, we construct a Cartesian grid $G_2(P)=G_2(S)$ of side length $\e2^{a}/8d$, and a point set $\alg_2(P)$. Again by P1 and P2, we have that  $\alg_2(P)=\alg_2(S)$. Thus, we output $\alg(P)=\alg_2(P)=S$, which finishes the proof.
\end{proof}

\section{PTAS for Stochastic Minimum $k$-Center}
\label{app:core}

Given an instance $\boldS=\{S_i\mid S_i\in \boldR^d, 1\leq i\leq N\}$ of a generalized $k$-median problem in $\R^d$ with a weight function $w:\boldS\rightarrow \R^+$, we show how to enumerate polynomially many sub-collections $\calS_i\subseteq \boldS$ together with their weight functions, such that there exists a generalized $\e$-coreset of $\boldS$. Recall that $\sigma_{\boldS}(S_i)$ is the sensitivity of $S_i$, and $\mathfrak{G}_{\boldS}=\sum_{i\in [N]}\sigma_{\boldS}(S_i)$ is the total sensitivity. Also recall that $\dim(\boldS)$ is the generalized dimension of $\boldS$. Define $q(S_i)=\sigma_{\boldS}(S_i)+1/N$ for $1\leq i\leq M$, and define $q_{\boldS}=\sum_{1\leq i\leq N}q(S_i)$. Note that $q_{\boldS}= \mathfrak{G}_{\boldS}+1\leq 4k+4$ by Lemma~\ref{lm:totalsen2}. Our algorithm is as follows.

\begin{enumerate}
\item Let $M=O((\frac{q_{\boldS}}{\e})^2\dim(\boldS))$. Let $L=\frac{10}{\e}(\log M+\log N+\log k)$.
\item Enumerate all collections $\calS_i \subseteq \boldS$ of cardinality at most $M$. Note that we only need to enumerate at most $N^M$ collections.
\item For a collection $\calS\subseteq \boldS$, w.l.o.g., assume that $\calS=\{S_1,S_2,\cdots, S_m\}$ ($m\leq M$). Enumerate all sequences $\bigl((1+\e)^{a_1},\ldots,(1+\e)^{a_m}\bigr)$ where each $0\leq a_i\leq L$ is an integer.
\item Given a collection $\calS=\{S_1,S_2,\cdots, S_m\}$ and a sequence $\bigl((1+\e)^{a_1},\ldots,(1+\e)^{a_m}\bigr)$, we construct a weight function $w':\calS\rightarrow \R^+$ as follows: for a point set $S_i\in \calS$, denote $w'(S_i)$ to be $(1+\e)^{a_i} \cdot w(S_i)/M $. Recall that $w(S_i)$ is the weight of $S_i\in \boldS$.
\end{enumerate}

\topic{Analysis.}
Recall that given an instance $\calP$ of a stochastic minimum $k$-center problem, we first reduce to an instance $\boldS=\alg(\calP)$ of a generalized $k$-median problem. Note that the cardinality of $\boldS$ is at most $n^{O(k/\e^d)}$, and the cardinality of a generalized $\e$-coreset is at most $M=O(\e^{-(d+2)}dk^4)$ by Theorem~\ref{thm:existpras}. Thus, we enumerate at most $N^M=n^{O(\e^{-(2d+2)}dk^5)}$ polynomially many sub-collections $\calS_i \subseteq \boldS$. For each collection $\calS_i$, we construct at most $M^{L+1}=O(n^{O(k/\e^d)})$ polynomially many weight functions. In total, we enumerate $N^M\cdot M^{L+1}=O(n^{O(\e^{-(2d+2)}dk^5)})$ polynomially many weighted sub-collections.

It remains to show that there exists a generalized $\e$-coreset of $\boldS$. We first have the following lemma.

\begin{lemma}
\label{lm:corestr}
Given an instance $\boldS=\{S_i\mid S_i\in \boldR^d, 1\leq i\leq N\}$ of a generalized $k$-median problem in $\R^d$ with a weight function $w:\boldS\rightarrow \R^+$, there exists a generalized $\e$-coreset $\calS\subseteq \boldS$ with a weight function $w':\calS\rightarrow \R^+$, such that
$$
\sum_{S\in \calS}w'(S)\cdot \maxdist(S,F)\in (1\pm \e)\sum_{S\in \boldS}w(S)\cdot \maxdist(S,F).
$$
The cardinality of $\calS$ is at most $M=O((\frac{q_{\boldS}}{\e})^2\dim(\boldS))$. Moreover, each weight $w'(S)$ ($S\in \calS$) has the form that $w'(S)=\frac{c \cdot q_{\boldS}\cdot w(S)}{q(S)\cdot M}$, where $1\leq c\leq M$ is an integer.
\end{lemma}

\begin{proof}
For each $S\in \boldS$, let $g_{S}: \calF\rightarrow \R^+$ be defined as $g_S(F)=w(S)\cdot \maxdist(S,F)/q(S)$. Let $D=\{g_S\mid S\in \boldS\}$ be a collection, together with a weight function $w'':D\rightarrow \R^+$ defined as $w''(g_S)=q(S)$. Note that for any $k$-point set $F\in \calF$, we have that
$$
\sum_{g_S\in G}w''(g_S)\cdot g_S(F)=\sum_{S\in \boldS}w(S)\cdot \maxdist(S,F)=\maxdist(\boldS,F).
$$
By Theorem 4.1 in \cite{FL11}, we can randomly sample (with replacement) a collection $\calS\subseteq D$ of cardinality at most $M=O((\frac{q_{\boldS}}{\e})^2\dim(\boldS))$, together with a weight function $w':\calS\rightarrow \R^+$ defined as $w'(g_S)= q_{\boldS}/M$. Then the multi-set $\calS$ satisfies that for every $F\in \calF$,
$$
\sum_{g_S\in \calS} w'(g_S) \cdot g_S(F) \in (1\pm \e)\sum_{g_S\in G}w''(g_S)\cdot g_S(F)=(1\pm \e)\maxdist(\boldS,F).
$$
By the definition of $g_S$ and $w'$, we prove the lemma.
\end{proof}

We are ready to prove the following lemma.

\begin{lemma}
\label{lm:coreext}
Among all sub-collections $\calS\subseteq \boldS$ of cardinality at most $M=O((\frac{q_{\boldS}}{\e})^2\dim(\boldS))$, together with a weight function $w':\calS\rightarrow \R^+$ of the form $w'(S_i)=(1+\e)^{a_i}\cdot w(S_i)/M$ ($0\leq a_i\leq \frac{10(\log M+\log N+\log k)}{\e}$ is an integer), there exists a generalized $\e$-coreset of $\boldS$.
\end{lemma}

\begin{proof}
By Lemma~\ref{lm:corestr}, there exists a generalized $\e$-coreset $\calS\subseteq \calS$ of cardinality at most $M$ together with a weight function $w':\calS\rightarrow \R^+$ defined as follows: each weight $w'(S)$ ($S\in \calS$) has the form that $w'(S)=\frac{c_S \cdot q_{\boldS}\cdot w(S)}{q(S)\cdot M}$ for some integer $1\leq c_S\leq M$. W.l.o.g., we assume that $\calS=\{S_1, S_2, \ldots, S_m\mid S_i\in \boldS\}$  ($m\leq M$).

By the definition of $q(S)$, we have that
$1/N\leq q(S)\leq q_{\boldS}= \mathfrak{G}_{\boldS}+1\leq 4k+4$.
Then we conclude that for each $S\in \calS$,
$$
1\leq \frac{c_S\cdot q_{\boldS}}{q(S)}\leq (4k+4)MN.
$$
For $1\leq i\leq m$, let $a_i=\lfloor\log_{1+\e} (\frac{c_{S_i}\cdot q_{\boldS}}{q(S_i)})\rfloor$. Note that each $a_i$ satisfies that $0\leq a_i\leq \frac{10(\log M+\log N+\log k)}{\e}$. Thus, we have enumerated the following sub-collection $\calS=\{S_1, S_2, \cdot, S_m\mid S_i\in \boldS\}$ with a weight function $w'':\calS\rightarrow \R^+$, such that $w''(S_i)=(1+\e)^{a_i}\cdot w(S_i)/M$. Moreover, for any $k$-point set $F$, we have the following inequality.
\begin{align*}
\sum_{1\leq i\leq m} w''(S_i)\cdot \maxdist(S_i,F)&=\sum_{1\leq i\leq m} \frac{(1+\e)^{a_i} \cdot w(S_i)}{M} \cdot \maxdist(S_i,F)\in (1\pm \e)\sum_{1\leq i\leq m} \frac{c_S \cdot q_{\boldS}\cdot w(S)}{q(S_i)\cdot M}\cdot \maxdist(S_i,F) \\
&= (1\pm \e)\sum_{1\leq i\leq m} w'(S_i)\cdot \maxdist(S_i,F) \in (1\pm 3\e) \sum_{S\in \boldS} w(S)\cdot \maxdist(S,F).
\end{align*}
The last inequality is due to the assumption that the sub-collection $\calS$ with a weight function $w'$ is a generalized $\e$-coreset of $\boldS$. Let $\e'=\e/3$, we prove the lemma.
\end{proof}

\end{document}